\documentclass[11pt,twoside]{article}
\usepackage{JDLL}
\usepackage{slashed}
\usepackage{array}
\usepackage{mathpartir}
\usepackage{tikz-cd}
\usepackage{lmodern}
\usepackage[titletoc,title]{appendix}
\usepackage{stmaryrd}
\usepackage{tikz}

\usepackage[round]{natbib}

\newcommand{\desk}{\top}

\usepackage[latin1]{inputenc}

\DeclareSymbolFont{bbold}{U}{bbold}{m}{n}
\DeclareSymbolFontAlphabet{\mathbbold}{bbold}

\newcolumntype{x}[1]{%
>{\centering\hspace{0pt}}p{#1}}%

\begin{document}

%Title
\title{\huge{Jets and differential linear logic}}
\author{James Wallbridge
}
\date{}
\maketitle
%Title

\begin{center}
\textbf{Abstract}
\end{center}

We prove that the category of vector bundles over a fixed smooth manifold and its corresponding category of convenient modules are models for intuitionistic differential linear logic.  The exponential modality is modelled by composing the jet comonad, whose Kleisli category has linear differential operators as morphisms, with the more familiar distributional comonad, whose Kleisli category has smooth maps as morphisms.  Combining the two comonads gives a new interpretation of the semantics of differential linear logic where the Kleisli morphisms are smooth local functionals, or equivalently, smooth partial differential operators, and the codereliction map induces the functional derivative.  This points towards a logic and hence computational theory of non-linear partial differential equations and their solutions based on variational calculus.

%Contents
\tableofcontents
%Contents

\newpage

%INTRODUCTION
\section{Introduction}

In this paper we study differential linear logic \citep{Eh3} through the lens of the category of vector bundles over a smooth manifold.  We prove that a number of categories arising from the category of vector bundles are models for intuitionistic differential linear logic.  This is part of a larger project aimed at understanding the interaction between differentiable programming, based on differential $\lambda$-calculus \citep{ER}, and differential linear logic with a view towards extending these concepts to a language of non-linear partial differential equations.  Since morphisms come from proofs in differential linear logic, and proofs are identified with programs in differential $\lambda$-calculus (via the Curry-Howard correspondence), the denotational semantics here provide tools for differentiable programming.

From a machine learning perspective, the work here suggests the possibility of a non-linear ``differential equation search" using gradient based optimization given some input and output boundary conditions in analogy with continuous ``function search", for example, searching among a subspace of the space of functional programs to find a program satisfying given constraints.  A particular case of the latter is neural architecture search given input-output data pairs.  This data driven programming has recently received attention in relation to the optimization of neural networks with various approaches to making the constituent functional blocks ``smooth" (see \citep{Gr,Zo,Ph} for selected works).  More recently, such tools have been used to solve certain partial differential equations with initial steps towards equation search (see \citep{La,We,Lo} for selected works).  

A more foundational approach to these questions, in the spirit of this paper, was recently proposed in \citep{Ke}.  With an eye towards non-linear phenomenon arising in a diverse range of scientific applications, we move beyond vector spaces to families of vector spaces parametrized by a smooth manifold.

There exist a number of approaches to the categorical semantics of differential linear logic in the literature.  These include K\"othe sequence spaces \citep{Eh1}, finiteness spaces \citep{Eh2}, convenient vector spaces \citep{BET} and vector spaces themselves \citep{CM}.  Our approach begins by considering a smooth generalization of \citep{CM} where our underlying objects are vector spaces parametrized by a fixed base manifold $M$.  More precisely, to a formula $A$ in differential linear logic, we associate the sheaf $\sE$ of sections of a vector bundle $E$ on $M$.  When $E$ is the trivial line bundle, then the associated denotation is simply the sheaf $\sC^\infty_M$ of smooth functions on $M$. 

We prove that there are two natural comonads on the category of vector bundles to model the exponential modality of linear logic.  Firstly, there is the jet comonad $!_j$ introduced in \citep{Mar} which sends a sheaf $\sE$ to the sheaf $!_j\sE$ of infinite jets of local sections of $E$.  An element of the exponentiation of a formula is an equivalence class of sections of a vector bundle with the same Taylor expansion at each point of $M$.  The idea of a syntactic Taylor expansion in linear logic and $\lambda$-calculus through the exponential connective \citep{ER,ER2} is therefore explicity present here in the semantics of vector bundles.  Working in the general setting of infinite jets, as opposed to $r$-jets for a fixed $r\in\mathbb{N}$, forces us to work in the enlarged category of pro-ind vector bundles \citep{GP}.  The objects in this category are (co)filtered objects in the category of vector bundles on $M$.  

In fact, to leverage better formal and functional analytic properties of vector bundles, especially in relation to dual objects, we move from pro-ind vector bundles to the category of convenient $\sC^\infty_M$-modules.  These are $\sC_M^\infty$-module objects in the category of sheaves of convenient vector spaces on $M$.  Sheaves of convenient vector spaces (see \citep{FK,KM} for the theory of convenient spaces) are a class of sheaves of infinite dimensional vector spaces which are general enough to include sections of an arbitrary vector bundle on a smooth manifold but which also retain excellent formal properties.  For example, the category of such objects is complete, cocomplete and closed symmetric monoidal.  The jet comonad $!_j$ descends to this category. 

The jet construction makes direct contact with the theory of linear differential operators and linear partial differential equations which enables us to understand these concepts within the setting of differential linear logic.  The Kleisli category for the jet comonad is the category of convenient vector bundles $\sE$ on $M$ and whose morphisms $!_j\sE\ra\sE'$ are linear differential operators.  This category is equivalent to the category of infinitely prolongated linear partial differential equations with $M$ as its manifold of independent variables.  Equivalently, as is always the case for the Kleisli category of a comonad, the objects are cofree $!_j$-coalgebras.  We prove that the category of convenient $\sC^\infty_M$-modules with the jet comonad is a symmetric monoidal storage category in the sense of \citep{BCLS}.  
 
The second comonad we consider will be called the distributional comonad.  Providing a symmetric monoidal storage category, which is moreover additive in an appropriate sense, with a codereliction map which models the differential in the context of linear logic, defines a model for intuitionistic differential linear logic.  It has been shown in \citep{BET} that the category of convenient vector spaces is a model for intuitionistic differential linear logic where the comonad is the map sending a convenient vector space to the Mackey-closure of the linear span of its Dirac distributions.  When the convenient vector space is finite dimensional, this is simply the space of distributions with compact support.  We extend this result to the category of convenient $\sC^\infty_M$-modules and continuous linear morphisms.  

The Kleisli category of the distributional comonad $!_\delta$ is the category of convenient $\sC^\infty_M$-modules and smooth morphisms.  The codereliction
\[   \bar{d}^\delta_\sE:\sE\ra !_\delta\sE  \] 
sends a section $s$ to $\lim_{h\ra 0}\frac{\delta_{hs}-\delta_0}{h}$.  The differential of a smooth functional $F:\sE\ra\sE'$ is then the linear map
\[   \tu{d}F:\sE\multimap(\sE\Rightarrow\sE')  \]
sending $(s,t)$ to the functional derivative $\tu{d}F(s,t)$ of $s$ in the direction $t$.  A familiar example is when $\sE$ and $\sE'$ are both the sheaf $\sC^\infty_M$ of smooth functions.  Then $F$ is an element of the continuous linear dual $(\sC^\infty_M)^\bot$.  Using the canonical evaluation pairing, this sheaf is isomorphic to the sheaf of compactly supported distributional densities on $M$.  

Combining the two comonads $!_j$ and $!_\delta$, which are proved to compose in the appropriate sense, becomes quite powerful.  We prove that the category of convenient $\sC^\infty_M$-modules with the composite comonad $!_\delta\circ!_j$ is a model for intuitionistic differential linear logic.  The codereliction
\[   \bar{d}^{j\delta}_\sE:\sE\ra !_{j\delta}\sE  \] 
sends a section $s$ to $\lim_{h\ra 0}\frac{\delta_{h(j(s))}-\delta_0}{h}$.  In this case we have a logic of smooth local functionals where a functional is local if and only if the value of its variables at a point $x$ in $M$ depends only on its infinite jet at that point.  These functionals are also known as Lagrangians and the functional derivative of a Lagrangian $L$ encodes the Euler-Lagrange equations (plus a total derivative).  The functional equation $\tu{d}L=0$ then encodes the space of solutions to the equations of motion.  Morphisms $!_\delta!_j\sE\ra\sE'$ from the Kleisli category are interpreted as smooth differential operators.  The interaction between these comonads shows how to pass between linear and non-linear objects.  More work is needed to understand the logic rules underlying this structure.  For linear partial differential equations with constant coefficients, this has been explored in \citep{Ke}. 

We end the paper by discussing how the above structure arises in the case where our vector bundle denotation is the trivial line bundle and our local functional is the Lagrangian for a free or self-interacting scalar field on an aribitrary Riemannian manifold.  In this case its convenient $\sC^\infty_M$-module is the sheaf of smooth functions on the manifold and the variational calculus leads to the space of solutions to the scalar field equations.  

\begin{rmk}
Most categories arising in linear logic \citep{Gi2}, objects of which include function spaces, are non-reflexive, ie. there is no canonical isomorphism between an object and its double dual.  This is also the case in our examples.  In \citep{Gi1}, Girard explored the denotational semantics of classical linear logic using the notion of coherent Banach space.  Adding coherence solves the issue of obtaining a monoidal category of reflexive objects.  However, one of the shortcomings of that model is a natural closed structure.  Another way of saying this is that coherent Banach spaces do not form a $*$-autonomous category.  Moreover, one cannot take the $*$-autonomous completion of the category of Banach spaces since they themselves do not form a closed symmetric monoidal category.  One can extend our results to the setting of classical differential linear logic by taking the $*$-autonomous completion of the closed symmetric monoidal category of convenient $\sC^\infty_M$-modules.  This completion, whose origins essentially go back to Mackey \citep{Ma}, is called the Chu-construction \citep{Ba} and is the universal way of overcoming the problem of reflexivity.  This remedy pairs each space with a subspace of its continuous linear dual in order to obtain the required canonical isomorphism $\sE\simeq\sE^{\bot\bot}$.  
\end{rmk}

%RELATION TO OTHER WORK
\section*{Relation to other work}

We consider vector bundles over a fixed base manifold $M$.  In the case of the distributional comonad with $M=*$, our results correspond to those of \citep{BET}, ie. the category of convenient $\sC^\infty_M$-modules reduces to the category of convenient vector spaces and the models for intuitionistic differential linear logic agree.  

When introducing the jet comonad, we have an interpretation of differential operators and linear partial differential equations within the logic.  Similar structure in the form of linear partial differential operators with constant coefficients has recently been studied in \citep{Ke} in the case where $M=*$ and the fiber over $*$ is Euclidean space $\mathbb{R}^n$.  

Combining the two comonads introduces a logical interpretation of non-linear differential operators.  A more general extension to non-linear cases, including the full theory of non-linear partial differential equations, involves considering morphisms of fibered manifolds.  To obtain a closed symmetric monoidal category in this setting requires moving outside the category of fibered manifolds and taking advantage of topos-theoretic and homotopical methods.  We will not consider this extension here.  However, see \citep{KS} for the theory of non-linear partial differential equations in more general synthetic categories.

%ACKNOWLEDGEMENTS
\section*{Acknowledgements}

The author would like to thank Kazuo Yano and Daniel Murfet for comments.  We also thank the anonymous referees whose comments led to significant improvements to the paper.

%MODELS FOR INTUITIONISTIC DIFFERENTIAL LINEAR LOGIC
\section{Models for intuitionistic differential linear logic}\label{midll}

In this section we recall what it means for a category to be a model for intuitionistic differential linear logic \citep{Eh3}.  To motivate such a definition, we recall some of its basic features.  We emphasise that this is not a complete presentation of the logic.  We merely highlight some properties in order to orient the reader towards the categorical definition at the end of this section.  

The syntax for intuitionistic linear logic involves the connectives $\{ \times, \otimes, !, \multimap \}$ with formulas $A$ generated by expressions of the form
\[  A ::= \desk  \,\,|\,\, 1  \,\,|\,\, A\times B \,\,|\,\, A\otimes B  \,\,|\,\,  A\multimap B \,\,|\,\, !A    \]
where $\desk$ and $1$ are units for $\times$ and $\otimes$ respectively.

Let $\Gamma$ and $\Theta$ be a (possibly empty) sequence of formulas $A_1,\ldots, A_n$.  The connectives satisfy various rules, which can be split between logic rules and structural rules.  The logic rules include, among others, the rules
\begin{mathpar}
   \inferrule*[right=$\times$]{\Gamma, A_1, A_2,\Theta \vdash B}{\Gamma, A_1\times A_2, \Theta\vdash B}
\quad\quad
\inferrule*[right=$\otimes$]{\Gamma, A_1, A_2,\Theta \vdash B}{\Gamma, A_1\otimes A_2, \Theta\vdash B}
\quad\quad
\inferrule*[right=$\multimap$]{\Gamma, A \vdash B}{\Gamma\vdash A\multimap B}
\end{mathpar}
for the additive, multiplicative and implicative connectives.  The structural rules are the exchange rule, identity rule, contraction rule, weakening rule and cut rule.  

Differential linear logic symmetrizes the contraction, weakening and dereliction rules  
\begin{mathpar}
  \inferrule*[right=${c}$]{\Gamma, !A, !A, \Theta\vdash B}{\Gamma, !A,\Theta\vdash B}
\quad\quad\quad
  \inferrule*[right=${w}$]{\Gamma,\Theta\vdash B}{\Gamma, !A, \Theta\vdash B}
\quad\quad\quad
  \inferrule*[right=${d}$]{\Gamma, A,\Theta \vdash B}{\Gamma, !A, \Theta\vdash B}
  \end{mathpar}
for the exponential of linear logic, by adding cocontraction, coweakening and codereliction rules
\begin{mathpar}
  \inferrule*[right=$\bar{c}$]{\Gamma, !A,\Theta\vdash B}{\Gamma, !A, !A, \Theta\vdash B}
\quad\quad\quad
  \inferrule*[right=$\bar{w}$]{\Gamma, !A, \Theta\vdash B}{\Gamma,\Theta\vdash B}
\quad\quad\quad
  \inferrule*[right=$\bar{d}$]{\Gamma, !A, \Theta\vdash B}{\Gamma, A,\Theta\vdash B}
\end{mathpar}
respectively.  This is a very natural thing to do in light of the symmetry inherent in the full classical linear logic \citep{Gi2} of which intuitionistic linear logic is a restriction thereof.  

Given a sequent $\Gamma\vdash B$, a proof of $\Gamma\vdash B$ is a series of sequents, beginning with basic axioms, and following various deduction rules which terminate with the sequent $\Gamma\vdash B$.  Two proofs are said to be \textit{equivalent} if they are equivalent under cut elimination.  

The goal of denotational semantics is to construct a category, a categorical semantics of differential linear logic in our case, faithfully reflecting this structure.  See \citep{Me} for an overview of the subject.  More generally, we should allow multicategories which, like categories, consist of a collection of objects, but allow multimorphisms from a finite sequence of objects to a single target object.  If we denote by $\llbracket A\rrbracket$ the denotation of a formula $A$, then $\llbracket A\rrbracket$ is an object of the multicategory whilst $\llbracket\Gamma\rrbracket\ra\llbracket B\rrbracket$ are multimorphisms for some collection of formulas $\Gamma$ \citep{HP}.  More precisely, the equivalence class of proofs of the sequent $\Gamma\vdash B$ under cut-elimination is assigned to the morphism.

The logical rules for each operation should follow from universal properties in the multicategory.  Working in a multicategory ensures that the connectives, together with their complete coherence data, satisfy a universal property.  For example, the tensor product in a general monoidal category is not universal.  It also illuminates the interpretation of the structural rules categorically.  Therefore, the discussion above should suggest, at minimum, the structure of a multicategory where the additive connective corresponds to the product and the multiplicative connective to tensor product.  

Forgoing some generality, we will work directly in a symmetric monoidal category $\mathsf{C}$ with finite products where the coherence data is contained in explicit proofs.  Note that we deliberately stay close to the notation used for connectives in linear algebra and make no distinction between the connectives in logic and those in the semantic model.  It should be clear from the context if we are referring to a multiplicative or additive product of formulas in logic or of objects in $\mathsf{C}$.

Before defining what we mean by a model for differential linear logic $\mathsf{C}$, we give an informal motivation for some of the various other structure on $\mathsf{C}$ which makes up the definition.  The denotation of implication in differential linear logic will correspond to an internal hom object
\[  \llbracket A\multimap B\rrbracket = \llbracket A\rrbracket\multimap\llbracket B\rrbracket := \uHom(\llbracket A\rrbracket,\llbracket B\rrbracket)  \]
making $\mathsf{C}$ a \textit{closed} symmetric monoidal category with finite products.  The contraction and weakening rules show that exponentiated objects $\llbracket !A\rrbracket$ in our category should satisfy coalgebraic rules, whilst the cocontraction and coweakening rules shows that algebraic rules should be satisfied.  In other words, we enter the realm of bialgebras in the categorical semantics of the exponentiated formulas.  

Finally, if we think of maps $\llbracket A\rrbracket\rightarrow\llbracket B\rrbracket$ as ``linear", then the map
\[  i:\Hom(\llbracket A\rrbracket,\llbracket B\rrbracket)\rightarrow\Hom(\llbracket !A\rrbracket, \llbracket B\rrbracket)   \]
given by composition with the dereliction map $d:\llbracket !A\rrbracket\rightarrow\llbracket A\rrbracket$ should be thought of as the inclusion of linear maps into ``non-linear" maps.  Correspondingly, composition with the codereliction $\bar{d}:\llbracket A\rrbracket\rightarrow\llbracket !A\rrbracket$ induces a map 
\[   D:\Hom(\llbracket !A\rrbracket,\llbracket B\rrbracket)\rightarrow\Hom(\llbracket A\rrbracket, \llbracket B\rrbracket)   \]
which is interpreted as a linearization map, ie. a differential operator.  Then $D[f]$ is the linear ``Jacobian" transformation.  The codereliction should satisfy conditions for $D$ to act like a differential operator, for example, satisfying the chain rule.

We now make this discussion more formal.  In doing so, we drop the bracket notation for the denotation of a formula for the remainder of this section.  The following is just a rewriting of the conditions found in Section~7 of \citep{BCLS} for a symmetric monoidal category with finite products to be a monoidal storage category (also called a (new) Seely category in the literature) \citep{Bi, Me, BCS2}.

\begin{dfn}\label{storagecomonad}
Let $(\mathsf{C},\otimes, 1)$ be a symmetric monoidal category with finite products $(\times,*)$.  A \textit{storage comonad} on $\mathsf{C}$ is a comonad $!=(!,\mu,\epsilon)$ on $\mathsf{C}$, where $\mu:!\ra !!$ is the comultiplication and $\epsilon:!\ra\id_\mathsf{C}$ is the counit, such that~: 
\begin{enumerate}
\item For all $A\in\mathsf{C}$, the object $!A$ is a cocommutative comonoid object in $\mathsf{C}$ with comultiplication $c_A:!A\ra !A\otimes !A$ and counit $e_A:!A\ra 1$ which are both natural transformations.
\item For all $A\in\mathsf{C}$, the map $\mu_A:!A\rightarrow !!A$ is a morphism of comonoid objects in $\mathsf{C}$.
\item For all $A,B\in\mathsf{C}$, the induced maps $e:!(*)\rightarrow 1$ and
\[     \big( !(\pi_1)\otimes!(\pi_2)\big) \circ c_{A\times B}:!(A\times B)\rightarrow !A\otimes !B       \]
are isomorphisms for the projection maps $(\pi_i)_{i\in\{1,2\}}$.
\end{enumerate}
\end{dfn}

The isomorphisms in Condition 3 of Definition~\ref{storagecomonad} are called \textit{Seely isomorphisms}.

\begin{dfn}
A \textit{symmetric monoidal storage category} is a symmetric monoidal category with finite products and a storage comonad.
\end{dfn}

Some models of logic can be endowed with $n$ linear exponential comonads for some $n\in\mathbb{N}$.  To compose comonads, we require extra structure for the composite to remain a comonad.  When $n=2$, we have the following definition from \citep{Bec}.

\begin{dfn}
Let $!_1=(!_1,\mu^1,\epsilon^1)$ and $!_2=(!_2,\mu^2,\epsilon^2)$ be two comonads on a category $\mathsf{C}$.  Then a \textit{distributive law} of $!_1$ over $!_2$ is a natural transformation  
$ \lambda:!_2\circ !_1\ra !_1\circ !_2 $
such that the following diagrams
\[
\begin{tikzcd}
!_2\circ !_1\circ !_1 \arrow[r, "\lambda\circ !_1"]   & !_1\circ !_2\circ !_1  \arrow[r,"!_1\circ\lambda "]   &  !_1\circ !_1\circ !_2\\
!_2\circ !_1  \arrow[u,"!_2\circ \mu^1"]  \arrow[rr,"\lambda"] \arrow[d, "\mu^2\circ !_1"] &     & !_1\circ !_2  \arrow[u,"\mu^1\circ !_2"]\arrow[d,"!_1\circ\mu^2"] \\
!_2\circ !_2\circ !_1  \arrow[r,"!_2\circ\lambda "] &  !_2\circ !_1\circ !_2    \arrow[r,"\lambda\circ !_2"]      &  !_1\circ !_2\circ !_2
\end{tikzcd}
\hspace{1cm}
\begin{tikzcd}
& !_1 &\\
!_2\circ !_1 \arrow[ur,"\epsilon^2\circ !_1"]\arrow[dr,"!_2\circ\epsilon^1"]\arrow[rr,"\lambda"]  &&  !_1\circ !_2 \arrow[ul,near end,"!_1\circ\epsilon^2"]\arrow[dl,"\epsilon^1\circ !_2"] \\
  & !_2      &
\end{tikzcd}
\]
commute in $\mathsf{C}$.
\end{dfn}

Let $!_1$ and $!_2$ be endofunctors on $\mathsf{C}$.  To simplify notation, set $!_{12}=!_2\circ !_1$.  If $!_1$ and $!_2$ are comonads and $\lambda$ a distributive law of $!_1$ over $!_2$, then there exists a unique composite comonad $!_{12}=(!_{12},\mu^{12},\epsilon^{12})$ where 
\[   \mu^{12}:!_{12}\xra{\mu^2\circ\mu^1}!_{22}\circ!_{11}\xra{!_2\circ\lambda\circ!_1} !_{12}\circ !_{12}  \]
is the comultiplication and 
\[   \epsilon^{12}:!_{12}\xra{\epsilon^2\circ\epsilon^1}\id_\mathsf{C} \]
the counit.  Note that the map $\lambda$ is left implicit in the notation for the composite comonad $!_{12}$.  The distributive law $\lambda$ lifts to a morphism of comonads if the diagrams
\[
\begin{tikzcd}
!_{12}  \arrow[d,"\mu^{12}"]\arrow[r,"\lambda"]        &      !_{21}\arrow[d,"\mu^{21}"] \\
!_{12}\circ !_{12} \arrow[r,"\lambda\circ\lambda"]    &   !_{21}\circ !_{21} 
\end{tikzcd}
\hspace{1cm}
\begin{tikzcd}[column sep=1.5em]
!_{12} \arrow[dr,"\epsilon^{12}"]\arrow[rr,"\lambda"]  &&  !_{21} \arrow[dl,near start,"\epsilon^{21}"] \\
  & \id_\mathsf{C}      &
\end{tikzcd}
\]
commute in $\mathsf{C}$.  If the distributive law is an isomorphism, then it is an isomorphism of comonads.

The following gives us the conditions on a category $\mathsf{C}$ for a composite comonad on $\mathsf{C}$ to be a storage comonad.

\begin{lem}\label{dchase}
Let $\mathsf{C}$ be a symmetric monoidal category with finite products.  If $!_1$ is a product preserving comonad, $!_2$ a storage comonad and $\lambda$ a distributive law of $!_1$ over $!_2$ on $\mathsf{C}$, then the composite comonad $!_{12}$ is a storage comonad on $\mathsf{C}$.
\end{lem}

\begin{proof}
From the conditions in the lemma, let $!_{12}$ be the unique composite comonad.  Consider the storage comonad $!_2$ acting on the object $!_1A$.  Then the object $!_{12}A$ is clearly a cocommutative comonoid object in $\mathsf{C}$ with comultiplication $c^{12}_A:!_{12}A\ra !_{12}A\otimes !_{12}A$ and counit $e^{12}_A:!_{12}A\ra 1$, and $\mu^{12}_A:!_{12}A\rightarrow !_{12}\circ !_{12}A$ a morphism of comonads in $\mathsf{C}$.  The Seely isomorphisms follow from the commutative diagrams
\[
\begin{tikzcd}
!_{12}(A\times B)  \arrow[d,"\sim"]\arrow[r,"c^{12}_{A\times B}"]        &    !_{12}(A\times B)\otimes !_{12}(A\times B) \arrow[d,"!_{12}(\pi_1)\otimes!_{12}(\pi_2)"] \\
!_{2}(!_1A\times !_1B) \arrow[r,"\sim"]    &   !_{12}A\otimes !_{12}B 
\end{tikzcd}
\hspace{1cm}
\begin{tikzcd}[column sep=1.5em]
!_{12}(*) \arrow[dr,"\sim"]\arrow[rr,"e"]  &&  1  \\
  & !_2(*)\arrow[ur,near start,"\sim"]&
\end{tikzcd}
\]
owing to the fact that $!_1$ preserves finite products, and $!_2$ is a storage comonad satisfying the Seely isomorphisms, respectively.
\end{proof}

Since many of the categories arising in applications are locally presentable \citep{Ad}, including our own, it is useful to include here another characterization of a symmetric monoidal storage category.

\begin{prop}\label{lps}
Let $\mathsf{C}$ be a locally presentable strong symmetric monoidal category with finite products.  Then $\mathsf{C}$ is endowed with a storage comonad if and only if there exists a locally presentable cartesian category $\mathsf{D}$ and a colimit preserving symmetric monoidal functor $L:\mathsf{D}\ra\mathsf{C}$ which is bijective on objects.
\end{prop}  

\begin{proof}
$(\Leftarrow)$ By the adjoint functor theorem, the functor $L$ admits a right adjoint $R:\mathsf{C}\ra\mathsf{D}$ and since $L$ is strong symmetric monoidal then $R$ is lax symmetric monoidal.  Thus $L\dashv R$ defines a linear-non-linear adjunction in the sense of \citep{Me}.  Since $L$ is moreover bijective on objects, by Proposition~25 of \textit{loc.cit}, the pair $(\mathsf{C},!:=L\circ R)$ define a symmetric monoidal storage category.  $(\Rightarrow)$ This follows from Proposition~24 of \textit{loc.cit.}.
\end{proof}

The category $\mathsf{D}$ in Proposition~\ref{lps} is isomorphic to the Kleisli category $\mathsf{C}_!$.  In the parlance of \citep{Me,Be}, we have a linear-non-linear adjunction which takes the form
\[ \begin{tikzcd}
\mathsf{C}_{!}\arrow[bend left]{r}{L}  
& \mathsf{C} \arrow[bend left]{l}{R}  
\end{tikzcd}  \]
for the comonad $!=L\circ R$ on $\mathsf{C}$. 

We need to introduce one more piece of structure which is key to interpreting a model for differential linear logic as a differential category \citep{BCS}, ie. a structure enabling one to ``differentiate" morphisms.  We will call a symmetric monoidal category $(\mathsf{C},\otimes)$ a \textit{$\CMon$-enriched symmetric monoidal category} if it is enriched over the monoidal category $(\textup{CMon},+)$ of commutative monoids such that the products are compatible in the sense that $(f+g)\otimes h = f\otimes h +  g\otimes h$ and $0\otimes h=0$ for zero morphisms $0$.

\begin{dfn}\label{cmonenriched}
A \textit{$\CMon$-enriched symmetric monoidal storage category}\footnote{These categories are called \textit{additive monoidal storage categories} in \citep{BCS,BCLS}.  We have decided to use the above more descriptive terminology and retain the standard use of additivity \citep{Mac}.} is a symmetric monoidal storage category which is also a $\CMon$-enriched symmetric monoidal category.
\end{dfn}

By Theorem~7.4 of \citep{BCLS}, $\CMon$-enriched symmetric monoidal storage categories have finite biproducts and an additive bialgebra modality.  So in addition to the cocommutative coalgebra $(!A,c_A,e_A)$, we have a commutative monoid object $(!A,\bar{c}_A,\bar{e}_A)$ for all $A\in\mathsf{C}$ with multiplication $\bar{c}_A:!A\otimes !A\ra !A$ and unit $\bar{e}_A:1\ra !A$.  The categorical analogue of the codereliction rule is then the following.

\begin{dfn}\label{codereliction}
Let $(\mathsf{C}, !)$ be a symmetric monoidal storage category which is also a $\CMon$-enriched category.  A natural transformation 
\[ \bar{d}:\id_\mathsf{C}\ra !  \] 
is called a \textit{codereliction} if it satisfies the rules~:
\begin{itemize}
\item (Constant rule) $e_A\circ\bar{d}_A = 0:A\ra 1$.
\item (Linear rule) $\epsilon_A\circ\bar{d}_A=\id_A:A\ra A$.
\item (Product rule) $c_A\circ\bar{d}_A = \bar{d}_A\otimes\bar{e}_A + \bar{e}_A\otimes\bar{d}_A:A\ra !A\otimes !A$.
\item (Chain rule) $\mu_A\circ\bar{c}_A\circ(\bar{d}_A\otimes\id_{!A}) = \bar{c}_{!A}\circ(\bar{d}_{!A}\otimes\mu_{A})\circ (\bar{c}_A\otimes\id_{!A})\circ(\bar{d}_A\otimes c_A):A\otimes !A\ra !!A$.
\end{itemize}
\end{dfn}

We now state the main overarching definition of this paper.

\begin{dfn}
A \textit{model for intuitionistic differential linear logic} is a $\CMon$-enriched symmetric monoidal storage category with a codereliction which is also a closed symmetric monoidal category.
\end{dfn}

One often finds the notion of a \textit{deriving transformation} \citep{Eh1,BCS} in place of a codereliction map in studies of differential categories.  Every codereliction induces a deriving transformation.  Furthermore, these two structures are equivalent on a $\CMon$-enriched symmetric monoidal storage category by combining Theorem~6 and Theorem~3 of \citep{BCLS}.  The deriving transformation associated to the codereliction $\bar{d}_A$ is given by the composition
\[  \overline\partial_A:!A\otimes A\xra{\id_{!A}\otimes\bar{d}_A}!A\otimes !A\xra{\bar{c}_A} !A  \]
in $\mathsf{C}$.  Then for any morphism $f:!A\ra B$ in $\mathsf{C}$, the composite map
\[   \tu{d}f:=f\circ\overline\partial_A:!A\otimes A\ra B  \]
will represent the \textit{derivative} of $f$ in $\mathsf{C}$.  

We define the $n$-fold derivative by induction~: we set $\overline\partial^0_A=\id_{!A}$ and
\[ \overline\partial^{n+1}_A:=\overline\partial_A\circ(\overline\partial^n_A\otimes\id_A):!A\otimes A^{\otimes n}\ra !A  \]
and define 
\[    \tu{d}^nf:=f\circ\overline\partial^n_A:!A\otimes A^{\otimes n}\ra B  \]
in $\mathsf{C}$.  The notation for intuitionistic implication $A\Rightarrow B:=!A\multimap B$ is now revealing since, by adjunction, the linear differential operator $\tu{d}^n$ is given by
\[    \tu{d}^n:(A\Rightarrow B)\ra (A^{\otimes n}\multimap(A\Rightarrow B))  \]
and should be thought of as sending a ``non-linear" morphism $f$ to a ``multi-linear" morphism $\tu{d}^nf:A^{\otimes n}\ra(A\Rightarrow B)$.  The basic example is the following.

\begin{ex}\label{diffex}
Given a smooth function $f:\mathbb{R}^n\ra\mathbb{R}^m$ in the category of vector spaces over $\mathbb{R}$, we have the linear morphism
\[ \tu{d}f:=f\circ\overline{\partial}_{\mathbb{R}^n}:\mathbb{R}^n\ra (\mathbb{R}^n\Rightarrow\mathbb{R}^m)  \]
given by $\tu{d}f(x)(y)=\tu{d}_xf(y)=(J_xf)y$ where $J_xf$ is the Jacobian of $f$ at $x$, ie. $\tu{d}f(x)(y)$ is the derivative of $f$ at $x$ in the direction $y$.  This satisfies the chain rule
\[ \tu{d}_x(g\circ f) = \tu{d}_{f(x)}(g)\circ\tu{d}_x(f) \]
in addition to other basic properties of differentiation contained in the definition of a codereliction.  
\end{ex}

In the following, Example~\ref{diffex} will be generalized to the case where $f$ is a section of an arbitrary vector bundle over a smooth manifold $M$ from the point of view of differential linear logic.  When $M$ is a point $*$ and the fibers are all of the form $\mathbb{R}^n$ for some $n\in\mathbb{N}$, we recover the simple example above.

%VECTOR BUNDLES AND THE JET COMONAD
\section{Vector bundles and the jet comonad}\label{jetcom}

We will henceforth work over the field $\mathbb{R}$ of real numbers and fix a smooth $n$-dimensional manifold $M$ over $\mathbb{R}$ for the remainder of the article.  

We consider the geometric approach to the theory of partial differential equations which begins with the study of jet bundles \citep{Sa}.  The $r$-jet of a function $f:M\ra\mathbb{R}$ at a point $x$ of $M$ can be thought of as the coordinate-free Taylor polynomial 
\[  (j^r_xf)(z):f(x) + f'(x)z +\ldots + \frac{1}{r!}f^{(r)}(x)z^r   \]
for a formal variable $z$.  We go beyond functions $f$, interpreted as sections of the trivial vector bundle $M\times\mathbb{R}\ra M$, and consider local sections of a general vector bundle in this paper.  The jet bundle associated to a vector bundle is itself a vector bundle whose coordinates represent the derivatives of the fiber coordinates.  

More precisely, let $\textup{VBun}(M)$ denote the category of finite rank vector bundles and $\pi:E\ra M$ an object in $\textup{VBun}(M)$.  To the vector bundle $E$, we associate its vector bundle $\pi_r:J^r(E)\ra M$ of $r$-jets of local sections.  For a local section $s$ at $x\in M$, its $r$-jet is denoted $j^r_x(s)$.  Two sections are in the equivalence class $j^r_x(s)$ if they have the same $r$th order Taylor expansion at $x$. 

Consider the sheaf $\s{C}^\infty_M$ of smooth functions on $M$ and its category $\Mod(\s{C}^\infty_M)$ of modules.  Then the functor $\textup{VBun}(M)\rightarrow\Mod(\s{C}^\infty_M)$ sending a vector bundle $E$ to its sheaf of sections $\s{E}:=\Gamma(E)$ is fully faithful with essential image the category $\cV(M)$ of locally free sheaves of finite rank.  We also refer to objects in this equivalent category as vector bundles on $M$.

The category $\cV(M)$ is a symmetric monoidal category in two ways.  Firstly, via the direct sum $\s{E}\oplus\s{E}'$ of sheaves, and secondly, via the tensor product $\s{E}\otimes_{\s{C}_M^\infty}\s{E}'$ of sheaves \citep{JP}.  We have canonical isomorphisms
\[  \s{E}\otimes_{\s{C}_M^\infty}\bigoplus_i\s{E}'_i\simeq\bigoplus_i(\s{E}\otimes_{\s{C}_M^\infty}\s{E}'_i)             \]
showing that the tensor product distributes over coproducts.  If we denote by $\uHom(\s{E},\s{E}')$ the sheaf of morphisms between $\s{E}$ and $\s{E}'$ which sends $U$ to $\Hom_{\s{C}^\infty_M|_U}(\s{E}|_U,\s{E}'|_U)$, then there exists an isomorphism 
\[  \Hom_{\s{C}^\infty_M}(\s{F}\otimes_{\s{C}^\infty_M}\s{E},\s{E}')\simeq\Hom_{\s{C}^\infty_M}(\s{F},\uHom(\s{E},\s{E}'))  \]
which makes the category of vector bundles $(\cV(M),\otimes,\sC^\infty_M)$ a closed additive symmetric monoidal category for the tensor product.  We will also be concerned with the cocartesian monoidal structure $(\cV(M),\oplus,0)$ where $0$ is the constant sheaf with value $\{0\}$ which is a zero object of $\cV(M)$.  We have an isomorphism
\[  \s{E}\oplus\s{E}'\simeq\s{E}\times\s{E}'  \]
of sheaves. 

Let $\s{J}^r(\sE)$ denote the sheaf of sections of $J^r(E)$ on $M$.  Given a section $s$ of $\pi_r|_U$, the $r$-jet prolongation of $s$ is the smooth section
\[  j^r(s):U\ra J^r(E)  \]
of $\pi_r$ such that $j^r(s)(x)=j_x^r(s)$ for all $x$ in $U\subseteq M$.  Then $j^r:\sE\ra\sJ^r(\sE)$ is a morphism of sheaves of sets.  Consider the endofunctor
$  !_{j^r}:\cV(M)\ra\cV(M)  $
sending $\sE$ to $\sJ^r(\sE)$ and a morphism $f:\sE\ra\sE'$ to its $r$-jet prolongation $\sJ^r(f):\sJ^r(\sE)\ra\sJ^r(\sE')$ which elementwise sends $j^r(s)$ to $j^r(f\circ s)$.  We will often make the abuse of writing $s\in\sE$ for a local section in $\sE|_U$.

Recall that a category $I$ is said to be \textit{cofiltered} \citep{Gr2} if it is non-empty; for any pair of objects $i$ and $j$ in $I$, there exists an object $k$ together with morphisms $k\ra i$ and $k\ra j$; and for every pair of morphisms $f$ and $g$ with the same source and target, there exists a morphism $h$ such that $f\circ h=g\circ h$.  A cofiltered diagram in a category $\mathsf{C}$ is a functor $X:I\ra\mathsf{C}$ indexed by a cofiltered category.  The category $\tu{Pro}(\mathsf{C})$ of \textit{pro-objects} in $\mathsf{C}$ has cofiltered diagrams in $\mathsf{C}$ as objects, and for two objects $X:I\ra\mathsf{C}$ and $Y:I'\ra\mathsf{C}$, morphisms defined by $\Hom_{\tu{Pro}(\mathsf{C})}(X,Y):=\lim_{i'\in I'}\colim_{i\in I}\Hom(X_i,Y_{i'})$.

Let $\sE$ be a vector bundle on $M$.  We denote by 
\[  \sJ(\sE)=\underset{r\in\mathbb{N}}{``\lim"} (\sJ^r(\sE)) \]
the pro-object
\[  \cdots\ra\sJ^{r+1}(\sE)\xra{\pi_{r+1,r}}\sJ^r(\sE)\ra\cdots\ra\sJ^1(\sE)\xra{\pi_{1,0}}\sJ^0(\sE)=\sE  \]
in the category $\cV(M)$ of vector bundles on $M$.  Here $\pi_{r+1,r}:\sJ^{r+1}(\sE)\ra\sJ^r(\sE)$ is the canonical projection.  

Given a pro-vector bundle $\sE:I\ra\cV(M)$, the infinite jet bundle of $\sE$ is the pro-object $\sJ(\sE):\mathbb{N}^{op}\times I\ra\cV(M)$ given by $``\lim"_{r,i}(\sJ^r(\sE_i))$ in $\cV(M)$.  Then the infinite jet prolongation $j:\s{E}\ra\sJ(\sE)$ lifts to a morphism of pro-sheaves.  We have an induced endofunctor $!_j:\tu{Pro}(\cV(M))\ra\tu{Pro}(\cV(M))$ sending $\sE$ to $\sJ(\sE)$ induced from infinite prolongation.  

The dual of a pro-object in the category $\cV(M)$ of vector bundles is an \textit{ind-object}.  The category of ind-objects in $\cV(M)$ will be denoted $\Ind(\cV(M))=\Pro(\cV(M)^{op})^{op}$.  If $\sE$ is a vector bundle, then the dual $\sJ(\sE)^\bot:=\uHom_{\sC^\infty_M}(\sJ(\sE),\sC^\infty_M)$ is an ind-object in $\cV(M)$.  Here $\uHom_{\sC^\infty_M}$ denotes the sheaf of continuous linear maps.  When $\sE$ is an ind-object in $\cV(M)$, then $\sJ(\sE)$ is a pro-ind-object in $\cV(M)$.

\begin{dfn}
A \textit{pro-ind vector bundle} is an object in the category $\tu{PI}(M):=\Pro(\Ind(\cV(M)))$ of pro-ind-objects in $\cV(M)$.
\end{dfn}

We will often identify a vector bundle with its image under the fully faithful map
$  i:\cV(M)\ra\tu{PI}(M) $
where $\cV(M)\simeq\Pro(\cV(M))\cap\Ind(\cV(M))\subset\tu{PI}(M)$ is an equivalence of categories.  

\begin{rmk}
Pro and Ind objects are often used in practice as presentations of infinite dimensional objects \citep{Gr1,Gr2}.  For example, the category of vector spaces $\Vect$ is equivalent to the category $\Ind(\Vect^{\tu{fin}})$ of ind-objects in the category $\Vect^{\tu{fin}}$ of finite dimensional vector spaces.  Alternatively, the category $\Pro(\Aff_S^{\,\tu{fin}})$ of pro-objects in the category $\Aff_S^{\,\tu{fin}}$ of affine schemes of finite type over a quasi-separated base scheme $S$ is equivalent to the category $\Aff_S$ of all affine schemes.  
\end{rmk}

We have an endofunctor
\[   !_j:\tu{PI}(M)\ra\tu{PI}(M)  \]
on the category $\tu{PI}(M)$ given by infinite prolongation.

\begin{lem}\label{comonad}
The endomorphism $!_j$ is a comonad on the category $\tu{PI}(M)$.
\end{lem}

\begin{proof}
We have a natural comultiplication map $\mu^j:!_j\ra !_j !_j$ which object-wise $\mu^j_\sE:!_j\sE\ra !_j!_j\sE$ sends $j_x(s)$ to $j_x(j(s))$ and a natural counit map $\epsilon^j:!_j\ra\id$ which object-wise $\epsilon^j_\sE:!_j\sE\ra\sE$ sends $j_x(s)$ to $s(x)$.  The commutativity of the relevant diagrams can be easily verified.
\end{proof}

\begin{rmk}
Currently $!_j\sE$ and its continuous linear dual $!_j\sE^\bot$ are defined as formal filtered limits and colimits.  This will be remedied in Section~\ref{csdc} by introducing functional analytic tools.
\end{rmk}

\begin{rmk}
The observation that the infinite jet functor defines a comonad in the smooth setting goes back to \citep{Mar}.  A far reaching generalization, encompassing many examples, is contained in \citep{KS}.  
\end{rmk}

Let $\tu{PI}(M)_{!_j}$ denote the Kleisli category of the comonad $!_j$.  We have a linear-non-linear adjunction
\[ \begin{tikzcd}
\tu{PI}(M)_{!_j}\arrow[bend left]{r}{X}  
& \tu{PI}(M) \arrow[bend left]{l}{U}  
\end{tikzcd}  \]
where $!_j=X\circ U$.  The left adjoint $X$ sends a pro-ind-vector bundle $\s{E}$ to $!_j\s{E}$ and a morphism $F:!_j\s{E}\ra\s{E}'$ to $!_j(F)\circ\mu^j_\s{E}:!_j\sE\ra !_j\sE'$.  The right adjoint is an identity on objects and sends a morphism $G:\s{E}\ra\s{E}'$ to $G\circ\epsilon^j_\s{E} = \epsilon^j_{\sE'}\circ !_j(G):!_j\sE\ra\sE'$ as a morphism in $\tu{PI}(M)$.  The unit of the adjunction on $\s{E}$ is simply the morphism $\eta_{\s{E}}=\id_{!_j\sE}:\sE\ra !_j\sE$, and the counit is given by $\epsilon^j_\sE$.  

We now give several interpretations of $\tu{PI}(M)_{!_j}$ which includes the theory of linear differential operators, $D$-modules, $!_j$-coalgebras and linear partial differential equations.  Let $\sE$ and $\sE'$ be vector bundles on $M$ and $\tu{Diff}^r(\sE,\sE')$  the sheaf of linear partial differential operators.  It sends $U\subseteq M$ to the $\s{C}^\infty_M(U)$-module whose elements are morphisms $P_U:\sE(U)\ra\sE'(U)$ given by $\sum_{|\alpha|\leq r}a_\alpha\circ\partial_{\alpha}$ for any trivialization where $a_\alpha\in\Hom_{\s{C}^\infty_M(U)}(\sE(U),\sE'(U))$.  The functor 
\[  \tu{Diff}^r(\sE,-):\cV(M)\ra\Set  \] 
is representable by the vector bundle $\sJ^r(\sE)$.  The isomorphism 
\[  \Hom_{\cV(M)}(\sJ^r(\sE),\sE')\simeq\tu{Diff}^r(\sE,\sE')  \] 
is given by the map $F\mapsto \widehat{F}:=F\circ j^r$.  The ind-object $\tu{Diff}(\sE,\sE'):=``\colim"_{r\in\mathbb{N}}\tu{Diff}^r(\sE,\sE')$ given by the natural inclusions induce an isomorphism
\[    
\tu{Diff}(\sE,\sE')\simeq\underset{r\in\mathbb{N}}{``\colim"}\Hom_{\cV(M)}(\sJ^r(\sE),\sE')\simeq\Hom_{\tu{PI}(M)}(\underset{r\in\mathbb{N}}{``\lim"}\sJ^r(\sE),\sE')   
\]
and so $\tu{Diff}(\sE,\sE')$ is represented by $\sJ(\sE)$ in $\tu{PI}(M)$.  The result of this discussion is that we can identify the image of the functor $U:\cV(M)\subset\tu{PI}(M)\ra\tu{PI}(M)_{!_j}$ in the Kleisli category of $!_j$ with the category of vector bundles on $M$ with linear partial differential operators as morphisms.  

Let $\widehat{F}:\sE\ra\sE'$ be a $r$th order differential operator.  We associate to $\widehat{F}$ its corresponding bundle map $F:\sJ(\sE)\ra\sE'$ and vice-versa.  Given a $q$th order differential operator $\widehat{G}$ between $\sE'$ and $\sE''$, composition with $\widehat{F}$ is given by
\[  G\circ F:\sJ^{r+q}(\sE)\xra{\mu^{r,q}_{\sE}}\sJ^q\sJ^r(\sE)\xra{\sJ^q\circ F}\sJ^q(\sE')\xra{G}\sE''  \]
where $\mu^{r,q}_{\sE}$ is the injection sending $j_x^{r+q}(s)$ to $j_x^qj^r(s)$ (and so $\mu=\mu^{\infty,\infty}$).  When they are both of infinite order, we obtain a linear differential operator $\widehat{G}\circ \widehat{F}:\sE\ra\sE''$ and Kleisli composition is well defined.

The Kleisli category of the jet comonad has a natural interpretation in the language of $D$-modules \citep{Ka}.  This extension is as follows.  Let $\mathscr{D}_M(\sE,\sE)$ denote the sheaf of linear differential operators on $M$ and 
\[ \s{D}_M^\infty:=\tu{Diff}(\sC^\infty_M,\sC^\infty_M) \] 
the sheaf of linear differential operators between the sheaf of smooth functions.  This is a sheaf of non-commutative $\sC_M^\infty$-algebras with product given by composition.  We denote the symmetric monoidal category of $D$-modules by 
\[  \Mod(\s{D}_M^\infty):=\Mod_{\s{D}_M^\infty}(\Mod(\sC^\infty_M))  \]
where the symmetric monoidal structure is given by tensoring over $\sC^\infty_M$.  If $\sE$ is a vector bundle, then endowing $\sE$ with a $D$-module structure is equivalent to the choice of flat connection
\[    \nabla:\sE\ra\Omega^1_M\otimes_{\sC^\infty_M}\sE  \]
on $\sE$ which characterizes $D$-modules with an underlying locally free $\sC^\infty_M$-module.  

The sheaf $\sJ(\sE)$ is endowed with a canonical $D$-module structure, the flat connection given by defining a section $\xi$ in $\sJ(\sE)$ to be flat if $\xi=j(s)$ for some $s\in\s{E}$, ie. horizontal sections of the connection are infinite prolongations of sections of $\sE$.  This is also called the Cartan connection.  Explicitly, after choosing coordinates $x_1,\ldots, x_n$ on $U\subseteq M$ and a trivialization $U\times E_0$ of $E$, we have
\[   \sJ(\sE)(U)=\sC^\infty_M(U)\otimes_{\mathbb{R}}\mathbb{R}[[x_1,\ldots,x_n]]\otimes_{\mathbb{R}}E_0  \]
and the flat connection is given by
$ \nabla(f\otimes g\otimes v)=\tu{d}f\otimes g\otimes v +\sum_i f\tu{d}x_i\otimes\frac{\partial}{\partial x_i}g\otimes v$.  Alternatively, it is defined through the Cartan distribution of tangent planes to sections of the form $j(s)$.  This is the map $\mu^{\infty,1}_\sE:\sJ(\sE)\ra\sJ^1\sJ(\sE)$ which is spanned by vector fields of the form 
\[  D_i=\frac{\partial}{\partial x_i}+\sum_{k,I}u^k_{Ii}\frac{\partial}{\partial u^k_I}  \] 
for fiber coordinates $u^k$ and a multi-index $I$.  Finally, there exists a bijection
\[  \Hom_{\sD^\infty_M}(\sJ(\sE),\sJ(\sE'))\simeq \tu{Diff}(\sE,\sE')  \]
which induces a fully faithful functor 
$   \s{J}:\tu{PI}(M)_{!_j}\ra\Pro(\Ind(\Mod(\s{D}_M^\infty)))  $
sending $\s{E}$ to $\sJ(\sE)$.

\begin{rmk}
For the multicategory interpretation of the Kleisli category one takes the multicategory of vector bundles and polydifferential operators 
\[  \tu{PolyDiff}(\sE_1\otimes\ldots\otimes\sE_n,\sE'):=\tu{Diff}(\sE_1,\sC^\infty_M)\otimes_{\sC^\infty_M}\ldots\otimes_{\sC^\infty_M}\tu{Diff}(\sE_n,\sC^\infty_M)\otimes_{\sC^\infty_M}\sE'     \]
where the action of $\sC^\infty_M$ on $\tu{Diff}(\sE_i,\sC^\infty_M)$ is given by left multiplication $(f D)(s):=f(Ds)$ for $s\in\sE_i$.  There exists a bijection
\[    \Hom_{\sD_M^\infty}(\sJ(\sE_1)\otimes_{\sC^\infty_M}\ldots\otimes_{\sC^\infty_M}\sJ(\sE_n),\sJ(\sE'))\simeq\tu{PolyDiff}(\sE_1\otimes\ldots\otimes\sE_n,\sE')  \]
where the left hand side denotes morphisms which are continuous.  
\end{rmk}

Another interpretation of $\tu{PI}(M)_{!_j}$ is as a full subcategory of the Eilenberg-Moore category of $!_j$-coalgebras.  A $!_j$-coalgebra for the comonad $!_j$ is a pair $(\sE,\nu_\sE)$ where $\sE$ is a pro-ind vector bundle and $\nu_\sE:\sE\ra!_j\sE$ is a morphism of pro-ind vector bundles such that $\epsilon_\sE\circ\nu_\sE=\id_\sE$ and $\mu_\sE\circ\nu_\sE=!_j(\nu_\sE)\circ\nu_\sE$.  A morphism between $!_j$-coalgebras $(\sE,\nu_\sE)$ and $(\sE',\nu_\sE')$ is a  morphism $f:\sE\ra\sE'$ of pro-ind vector bundles such that $!_j(f)\circ\nu_\sE=\nu_\sE'\circ f$.  The category of $!_j$-coalgebras, often called the Eilenberg-Moore category, will be denoted $\tu{PI}(M)^{!_j}$.

The Eilenberg-Moore category of $!_j$ is equivalent to a certain category of partial differential equations introduced in \citep{Vi} (see \citep{Mar}).  We first recall some geometric definitions \citep{Po}.

\begin{dfn}
Let $\pi:E\ra M$ be a vector bundle.  A $r$th order \textit{partial differential equation} on $E$ is a fibered submanifold of $\pi_r:J^r(E)\ra M$.  An \textit{inhomogeneous linear partial differential equation} is an affine subbundle of $\pi_r$.  A \textit{homogeneous linear partial differential equation} is a vector subbundle of $\pi_r$.  
\end{dfn}

Let $H^r$ be a $r$th order linear partial differential equation.  In the homogenous case, there exists a vector bundle $E'=\coker(H^r)$ on $M$ and a morphism of vector bundles $f:J^r(E)\ra E'$ such that $H^r=\ker(f)$.  This corresponds to the standard interpretation $f(x_i,u^\alpha,u^\alpha_I)=0$ where $u^\alpha$ are coordinates in the fiber of $E$.  A linear partial differential equation will be henceforth considered homogenous unless otherwise specified.  A  (local) solution of a $r$th order partial differential equation $H^r$ is a section $s$ of $\pi_r|_U$ such that $j^rs(x)\in H^r$ for all $x\in U$.  

The $q$th order prolongation of $h:H^r\subseteq J^r(E)$ is the pullback
\[
\begin{tikzcd}
H^{r,q}\arrow[r," "]  \arrow[d," "]  & J^{r+q}(E) \arrow[d,"\mu^{r,q}_E"] \\
J^q(H^r)\arrow[r,"J^q(h)"]   & J^q J^r(E)
\end{tikzcd} 
\]
in the category of vector bundles.  The infinite prolongation $H\subseteq J(E)$ of $H^r$ is the pro-object
\[   \cdots\ra {H}^{r,k+1}\xra{\pi^r_{k+1,k}}{H}^{r,k}\ra\cdots\ra {H}^{r,1}\xra{\pi^r_{0,1}}{H}^{r,0}={H}^r\subseteq J^r(E)  \]
in the category of vector bundles.  It can be interpreted as $H^r$ together with its system of total derivatives.  A morphism between infinitely prolongated linear equations is a morphism of pro-vector bundles.  

These constructions are clearly extended to the case where $E$ itself is a pro-ind-vector bundle.  We obtain a category $\tu{LPDE}(M)$ of infinitely prolongated linear partial differential equations.  

\begin{rmk}\label{symmetries}
In this geometric formulation of partial differential equations, infinitesimal symmetries are given by tangent vector fields on the jet bundle whose flows preserve this submanifold \citep{Po,Ol}.
\end{rmk}

The sheaf interpretation of this result is as follows.  The vector bundle $H^r$ induces a sheaf $\s{H}^r$ of solutions and $H^{r,q}$ a prolongated sheaf $\s{H}^{r,q}\subset\sJ^{q+r}(\sE)$ of solutions.  The infinite prolongation $\s{H}\subseteq\sJ(\sE)$ is a pro-object in the category of vector bundles $\cV(M)$ over $M$.  If $\sE$ is a pro-ind vector bundle, then the same is so for $\s{H}$.  There is an equivalence $\s{H}\simeq\s{H}^r$ of sheaves, ie. a section of $\sE$ is a solution of $H^r$ if and only if it is a solution of the prolonged equation $H$.

We call the map $h:\s{H}\ra !_j\sE$ simply the sheaf of solutions.  Given two sheaves of solutions $h:\s{H}\ra !_j\sE$ and $h':\s{H}'\ra !_j\sE'$, a morphism $\gamma:h\ra h'$ is a commutative diagram
\[
\begin{tikzcd}
\s{H}\arrow[r,"h"]  \arrow[d," "]  & !_j\sE \arrow[d," "] \\
\s{H}'\arrow[r,"h'"]   & !_j\sE'
\end{tikzcd} 
\]
in $\tu{PI}(M)$.  We denote by $\tu{Soln}(M)$ the category of sheaves of solutions and morphisms between them.  

\begin{prop}\label{solneq}
There exists a chain of equivalences
$   \tu{LPDE}(M)\simeq\tu{Soln}(M)\simeq\tu{PI}(M)^{!_j}    $
of categories.
\end{prop}

\begin{proof}
This can be deduced from Proposition~2.4 and Proposition~2.5 of \citep{Mar} so we only sketch the proof.  The first equivalence is clear.  For the second, consider the sheaf of solutions $h^r:\s{H}^r\ra !_{j^r}\sE$ to a $r$th order linear partial differential equation $H^r\subseteq J^r(E)$ and its corresponding infinite prolongation $h:\s{H}\ra !_j\sE$.  Consider the diagram
\[ \begin{tikzcd}[row sep=scriptsize, column sep=scriptsize]
& !_j\s{H}^r \arrow[rr, "!_j(h^r)"] \arrow[dd, near start, "\mu_{\s{H}^r}"] & & !_j!_{j^r}\sE  \arrow[dd, "\mu_{!_{j^r}\sE}"] \\ \s{H} \arrow[ur, "h^*"] \arrow[rr, crossing over, near end, "h"] \arrow[dd, near start, "\tilde{h}"] & & !_j\sE \arrow[ur, "\mu^{\infty,r}_\sE"]\\
& !_j!_j\s{H}^r  \arrow[rr,near start, "!_j!_j(h^r)"] & & !_j!_j !_{j^r}\sE  \\
!_j\s{H} \arrow[ur, "!_jh^*"]\arrow[rr, "!_j(h)"] & & !_j!_j\sE \arrow[ur, near end, "!_j\mu^{\infty,r}_{\sE}"] \arrow[from=uu, crossing over, near start, "\mu_\sE"]\\
\end{tikzcd} \]
in $\tu{PI}(M)$ where $h^*$ is the morphism making the square in the top face commute and $\tilde{h}$ is the morphism making the resulting full diagram commute.  We have a functor $\tilde{(\cdot)}$ sending the solution sheaf $h$ to the pair $(\s{H},\tilde{h}:\s{H}\ra !_j\s{H})$ and this pair can be shown to be a $!_j$-coalgebra.  The right adjoint functor sends a $!_j$-coalgebra $(\sE,\nu:\sE\ra!_j\sE)$ to the solution sheaf $\nu:\sE\subset !_j\sE$ satisfying $\mu_{\sE}=!_j(\nu)$ which is infinitely prolonged.  Then composition with $\tilde{(\cdot)}$ gives an adjoint equivalence.
\end{proof} 

There exists a natural inclusion
\[  \tu{PI}(M)_{!_j}\hookrightarrow\tu{PI}(M)^{!_j}  \]
sending a pro-ind vector bundle $\sE$ to $(!_j\sE,\mu^j_{\sE})$ and a differential operator $F:!_j\sE\ra\sE'$ to the composition $!_j(F)\circ\mu^j_\sE:!_j\sE\ra !_j\sE'$.  The essential image of this inclusion is the full subcategory of $!_j$-coalgebras spanned by cofree $!_j$-coalgebras.  This follows from the fact that the Kleisli category of any comonad is equivalent to the subcategory of cofree coalgebras of the comonad in the Eilenberg-Moore category.  Owing to  Proposition~\ref{solneq}, objects in $\tu{PI}(M)_{!_j}$ can be identified with the sheaf of solutions to a cofree infinitely prolongated linear partial differential equation.   

The category $\tu{PI}(M)$ is not a symmetric monoidal storage category with the monoidal structure given by the tensor product, since the comonad $!_j$ does not satisfy the Seely isomorphisms.  However, for the cocartesian monoidal structure, it is satisfied.

\begin{prop}\label{thm1}
The category of pro-ind vector bundles on $M$ with the jet comonad $!_j$ is a symmetric monoidal storage category for the cocartesian monoidal structure.
\end{prop}

\begin{proof}
The category $(\cV(M),\oplus,0)$ of smooth vector bundles is a $\CMon$-enriched symmetric monoidal category and so we can deduce that the category $\tu{PI}(M)$ pro-ind objects in $\cV(M)$ is also $\CMon$-enriched symmetric monoidal.  Let $\sE,\sE'\in\tu{PI}(M)$ and $s$ be a local section of $\sE$.  By Lemma~\ref{comonad}, $!_j$ is a comonad.  Since $\oplus$ is also a product, every object of $\tu{PI}(M)$ has a unique comonoid structure given by the diagonal map which is cocommutative.  Moreover, any morphism in $\tu{PI}(M)$ is automatically a comonoid morphism.  Therefore, for a comonoid $(!_j\sE,c_\sE,e_\sE)$, the comultiplication is given by
\[  {c}_\sE:!_j\sE\ra !_j\sE\times !_j\sE,  \]
the counit is given by $e_\sE:!_j\sE\ra 0$ which sends $j(s)$ to zero, and $\mu_\sE:!_j\sE\ra !_j!_j\sE$ is a morphism of comonoid objects.  Furthermore, the morphism
\[     \big( !_j(\pi_0)\oplus!_j(\pi_1)\big) \circ c_{\sE\times\sE'}:!_j(\sE\times\sE')\rightarrow !_j\sE\oplus !_j\sE'   \]
is an isomorphism in $\tu{PI}(M)$ since $!_j(\sE\times\sE')\simeq !_j(\sE\oplus\sE')$ and $j(s + s')\simeq j(s) + j(s')$.  Finally, the morphism
\[  e:!_j(*)\rightarrow 0   \]
is an isomorphism since the terminal object $*$ in $\tu{PI}(M)$ is the pro-ind zero vector bundle $0$ and it is clear that $!_j(0)\simeq 0$.  As a result, $!_j$ is a storage comonad and $\tu{PI}(M)$ is a symmetric monoidal storgage category.
\end{proof}

\begin{ex}(Connections).
In analogy with a codereliction, we introduce a map
\[  \Gamma^1_\sE:\sE\ra !_{j^1}\sE \]
in $\cV(M)$ which is natural in $\sE$, such that the linear rule $\epsilon^{j}_\sE\circ\Gamma^1_\sE=\id_\sE$ is satisfied for the comonad $!_{j^1}$.  This is simply a (linear) \textit{connection}.  Indeed, consider the canonical map $j^1:\sE\ra\sJ^1(\sE)$ of sheaves.  Elements in the kernel of this map can be written as $df\otimes s$.  The \textit{covariant derivative} associated to $\Gamma^1_\sE$ is then the (non $\sC_M^\infty$-linear) map
\[  \nabla:\sE\ra\Gamma(\Omega^1\otimes_{\mathbb{R}}E) \]
satisfying the Leibniz rule
\[  \nabla(fs) = f\nabla(s) + df\otimes s   \]
where $\Omega^1$ is the vector bundle of one-forms on $M$.  In local coordinates $(x_i,u^k,u^k_i)\circ\Gamma^1_\sE=(x_i,u^k,\Gamma_i^k)$, its local expression is $\nabla=dx^i\otimes(\partial_i+\Gamma^k_i\partial_k)$.  More generally, higher-order connections $\Gamma^k_\sE:!_{j^{k-1}}\sE\ra !_{j^k}\sE$ can be defined \citep{Li} .
\end{ex}

\begin{ex}(Tangent vector fields).
Let $E=TM$ be the tangent bundle and $\sE=\s{X}$ the sheaf of vector fields on $M$.  Consider the sequent $!A\vdash B$ in linear logic with denotation $\llbracket -\rrbracket_M$ given by a first-order map $F:\llbracket !A\rrbracket_M=!_{j^1}\s{X}\subset!_j\s{X}\ra \llbracket B\rrbracket_M=\sE'$.  Given a vector field $s:U\ra TU$ on $U\subset M$, we have the first-jet
\[ j^1(s):U\ra J^1(TU)\subset J(TU)  \]
to $s$ and a commutative diagram
\[
\begin{tikzcd}
 & !_{j}\s{X}_U \arrow[d,"\epsilon^j_{\s{X}_U}"]  \arrow[dr,"F_U"]  \\
* \arrow[ur, near end, "j^1(s)"] \arrow[r,"s"] & \s{X}_U   \arrow[r, "\widehat{F}_U "] & \sE'_U
\end{tikzcd}
\]
in $\tu{PI}(M)$.  Here $\widehat{F}_U:\s{X}_U\ra\sE'_U$, where $\widehat{F}_U(s)\simeq F_U(j^1(s))$, is the first-order linear differential operator associated to $F_U$.  
\end{ex}

%CONVENIENT SHEAVES AND THE DISTRIBUTIONAL COMONAD
\section{Convenient sheaves and the distributional comonad}\label{csdc}

Up until now, we have considered the category of pro-ind objects in $\cV(M)$.  However, there is another approach which takes advantage of functional analytic properties of the space of sections of a vector bundle.  In particular, the category of pro-ind-vector bundles have several poor formal properties arising from the category $\cV(M)$.  This can be remedied by embedding $\tu{PI}(M)$ into an appropriate category.  We accomplish this by endowing all our function spaces with a complete bornological structure \citep{HN}, or equivalently, a convenient vector space structure \citep{FK,KM}.   

There are a number of equivalent ways one can define the category of convenient vector spaces \citep{KM}.  Our choice is the following.  Let $\tu{Born}$ denote the category of (convex) bornological vector spaces and bounded linear morphisms and $\tu{LCTVS}$ the category of locally convex topological vector spaces and continuous linear morphisms.  Consider the adjunction
\[ \begin{tikzcd}
   \tu{Born}\arrow[bend left]{r}{\gamma}  
& \tu{LCTVS} \arrow[bend left]{l}{\beta}  
\end{tikzcd}  \]
where $\gamma$ is left adjoint to the functor $\beta$ associating to a locally convex topological vector space the bornological vector space with its von-Neumann bornology.  The functor $\gamma$ is fully faithful.  Therefore, we have an isomorphism $V\simeq\beta\circ\gamma(V)$ in $\tu{Born}$, ie. every bornological vector space is isomorphic to a vector space whose bornology comes from some locally convex topological vector space.  The equivalent category of topological bornological vector spaces will be denoted $\tu{TBorn}$.  A topological bornological vector space $V$ is said to be \textit{$c^\infty$-complete} if a curve $c:\mathbb{R}\ra V$ is smooth if and only if for every bounded linear functional $f:V\ra\mathbb{R}$, the composition $f\circ c:\mathbb{R}\ra\mathbb{R}$ is smooth.

We will define the category $\Conv$ of convenient vector spaces to be the full subcategory of $\tu{TBorn}$ spanned by $c^\infty$-complete objects.  The inclusion functor from $\Conv$ to the category $\tu{TBorn}$ has a left adjoint
\[   c^\infty:\tu{TBorn}\ra\Conv  \]
called the $c^\infty$-completion.

The category $\Conv$ is a closed symmetric monoidal category.  We will be careful to distinguish the structure on various function spaces.  For convenient vector spaces $V$ and $W$, then $\Hom(V,W)$ will denote the set of morphisms, $\Hom_\mathbb{R}(V,W)$ the $\mathbb{R}$-vector space of $\mathbb{R}$-linear morphisms and $\uHom(V,W)$ the convenient vector space of continuous $\mathbb{R}$-linear morphisms.  We use the notation $V^\vee:=\Hom_{\mathbb{R}}(V,\mathbb{R})$ for the linear dual of $V$ and $V^\bot:=\uHom(V,\mathbb{R})$ for the continuous linear dual.  

Let $\pi:E\ra M$ be a vector bundle on $M$.  For any $U\subseteq M$, we endow the vector space $\s{E}(U)$ of sections of $E$ with the structure of a convenient vector space induced from the nuclear Fr\'echet topology of uniform convergence on compact subsets in all derivatives seperately.  This makes $\s{E}$ a sheaf of convenient vector spaces on $M$.  The same holds for the cosheaf $\sE_c$ of compactly supported sections of $E$.  See Lemma~5.1.1 of \citep{CG} for a formal proof.

In particular, $\sC^\infty_M$ is a sheaf of convenient vector spaces and moreover a sheaf of convenient algebras.  An algebra is said to be \textit{convenient} if it is a commutative monoid object in the symmetric monoidal category $\Conv$.  This makes $\sE$ a $\s{C}^\infty_M$-module object in the category $\tu{Sh}_{\Conv_k}(M)$ of sheaves of convenient vector spaces.  The category of convenient $\s{C}^\infty_M$-modules will be denoted by
\[    \tu{ConMod}(\s{C}^\infty_M):=\Mod_{\s{C}^\infty_M}(\tu{Sh}_{\Conv_k}(M)).     \]
We have a fully faithful inclusion
\[   i:\tu{PI}(M)\ra\tu{ConMod}(\s{C}^\infty_M)  \]
of categories.  The inclusion sends a pro-ind vector bundle $``\lim"_{r\in\mathbb{N}}``\colim_{q\in\mathbb{N}}"\sE$ to the genuine limit $\lim_{r\in\mathbb{N}}\colim_{q\in\mathbb{N}}\sE$ in $\tu{ConMod}(\s{C}^\infty_M)$.  This limit is well defined since the category of convenient $\sC^\infty_M$-modules is complete and cocomplete.

The category $\tu{ConMod}(\s{C}^\infty_M)$ is a closed symmetric monoidal category with tensor product $\otimes_{\sC^\infty_M}$ which we simply denote by $\otimes$.  The $\sC^\infty_M$-module of continuous linear morphisms between two $\sC^\infty_M$-modules $\sE$ and $\sE'$ will be denoted $\uHom_{\sC^\infty_M}(\sE,\sE')$.  

We now describe some important examples of (co)sheaves of convenient spaces.  Let $\s{T}^{\infty}$ be the convenient sheaf of distributions on $M$ and denote by
\[  \overline{\s{E}}:=\s{E}\otimes_{\s{C}^\infty_M}\s{T}^{\infty}  \]
the convenient sheaf of distributional sections of $E$ on $M$.  Let $\s{T}_{c}^\infty$ be the convenient cosheaf of compactly supported distributions on $M$ and 
\[  \overline{\s{E}}_c:=\s{E}_c\otimes_{\s{C}_c^\infty}\s{T}_{c}^\infty  \] 
the convenient cosheaf of compactly supported distributional sections of $E$ on $M$.  We let $\tu{Dens}(M):=\wedge^nT^*M\otimes\mathfrak{o}_M$ denote the vector bundle of densities on $M$ where $\mathfrak{o}_M$ is the orientation line bundle and $\s{D}ens_M$ the convenient sheaf of sections of $\tu{Dens}(M)$.  

Let $\s{E}^\forall$ denote the convenient sheaf of sections of the vector bundle $E^\forall=E^\vee\otimes\tu{Dens}(M)$ on $M$ where $E^\vee$ is the fiberwise linear dual.  Likewise, let $\s{E}^\forall_c$ denote the convenient cosheaf of compactly supported sections of $E^\vee\otimes\tu{Dens}(M)$ on $M$.  We define $\s{E}^\bot:=\uHom_{\sC^\infty_M}(\s{E},\sC^\infty_M)$ and $\s{E}^\bot_c:=\uHom_{\sC^\infty_M}(\s{E}_c,\sC^\infty_M)$ to be the continuous linear duals endowed with the strong topology of uniform convergence on bounded subsets.  

The fiberwise evaluation pairing between $E$ and $E^\vee$ induces a morphism $fib(-,-):E^\forall\otimes E\ra\tu{Dens}(M)$ of vector bundles which extends to a pairing 
\[  ev_U:\s{E}^\forall_c(U)\times\s{E}(U)\ra\sC^\infty_M(U)  \]
of convenient $\sC^\infty_M(U)$-modules given by sending a pair $(\omega,s)$ on $U\subseteq M$ to the integral $\int_U fib(\omega,s)$.  This construction induces isomorphisms
\[  \s{E}^\bot(U)\simeq\overline{{\s{E}}^\forall_c}(U)   \quad\quad\quad   \s{E}^\bot_c(U)\simeq\overline{{\s{E}^\forall}}(U)  \] 
of convenient $\sC^\infty_M(U)$-modules.   

Let $V$ be a convenient vector space.  A curve $c:\mathbb{R}\ra V$ is said to be \textit{smooth} if all derivatives of $c$ exist in the underlying topological space of $V$.  The set of smooth curves in $V$ is denoted $\cC_V$.  A morphism $f:V\ra W$ of convenient vector spaces is said to be \textit{smooth} if $f(\cC_V)\subseteq\cC_{W}$.  Finally, a morphism $f:\sE\ra\sE'$ between convenient $\sC_M^\infty$-modules is \textit{smooth} if $\sE(U)\ra\sE'(U)$ is smooth for all $U\subseteq M$.  We denote by $\uHom^{\tu{sm}}_{\sC^\infty_M}(\sE,\sE')$ the $\sC^\infty_M$-module of smooth morphisms and 
\[  \s{E}^*:=\uHom^{\tu{sm}}_{\sC^\infty_M}(\s{E},\sC^\infty_M)  \] 
the smooth dual.

Let $\tu{ConMod}^{\tu{sm}}(\s{C}^\infty_M)$ denote the closed symmetric monoidal category of convenient $\sC^\infty_M$-modules and smooth morphisms.  We deduce from Corollary~2.11 of \citep{KM} that a linear map between convenient $\sC^\infty_M$-modules is smooth if and only if it is a bornological morphism.  Therefore, we have a natural forgetful functor
\[    U:\tu{ConMod}(\s{C}^\infty_M)\ra\tu{ConMod}^{\tu{sm}}(\s{C}^\infty_M)   \]
which is the identity on objects and forgets the linear structure.  

We now define a number of different functionals on the space of sections of a vector bundle.

\begin{dfn}
Let $E$ be a vector bundle on $M$.  A \textit{linear functional} on $\sE$ is an element of the continuous linear dual $\sE^\bot$.  A \textit{smooth functional} on $\sE$ is an element of the smooth dual $\sE^*$.  
\end{dfn}

\begin{ex}[Polynomial functions]\label{poly}
An intermediate class of smooth functionals are polynomials.  The algebra of polynomial functions on $\sE$ is given by
\[   \s{O}_{\s{E}}:={\Sym}_{\sC^\infty_M}(\s{E}^\bot)=\bigoplus_{n=0}^\infty((\s{E}^\bot)^{\otimes n})_{S_n}\simeq\bigoplus_{n=0}^\infty((\overline{\s{E}^\forall_c})^{\otimes n})_{S_n}   \]
where the subscript $S_n$ refers to taking coinvariants with respect to the action of the symmetric group on the $n$-fold tensor product.  The algebra of polynomial functions on $\sE_c$ is given by 
\[   \s{O}_{\s{E}_c}:={\Sym}_{\sC^\infty_M}(\s{E}_c^\bot)=\bigoplus_{n=0}^\infty((\s{E}_c^\bot)^{\otimes n})_{S_n}\simeq\bigoplus_{n=0}^\infty((\overline{\s{E}^\forall})^{\otimes n})_{S_n}.   \]
\end{ex} 

\begin{ex}[Formal power series]\label{fps}
A larger class of smooth functionals are those given by formal power series.  That is, the completed symmetric algebra
\[   \widehat{\s{O}}_{\s{E}}:=\widehat{\Sym}_{\sC^\infty_M}(\s{E}^\bot)=\prod_{n=0}^\infty((\s{E}^\bot)^{\otimes n})_{S_n}\simeq\prod_{n=0}^\infty((\overline{\s{E}^\forall_c})^{\otimes n})_{S_n}   \]
and that on compactly supported sections 
\[   \widehat{\s{O}}_{\s{E}_c}:=\widehat{\Sym}_{\sC^\infty_M}(\s{E}_c^\bot)=\prod_{n=0}^\infty((\s{E}_c^\bot)^{\otimes n})_{S_n}\simeq\prod_{n=0}^\infty((\overline{\s{E}^\forall})^{\otimes n})_{S_n}.   \]
This leads to natural inclusions $\sE^\bot\subset\sO_{\sE}\subset\widehat{\sO}_\sE\subset\sE^*$ of sheaves and similarly for compactly supported sections.  See \citep{KT} for a detailed discussion of polynomials and power series in a similar context.
\end{ex} 

We now describe a comonad which we call the distributional comonad which is a generalization of that contained in \citep{BET} to the setting of $\sC_M^\infty$-modules.  Consider the Dirac distributional density map
\[   \delta:\sE\ra(\sE^*)^\bot  \]
sending a section $s$ to $\delta_s:F\mapsto F(s)$ where $F$ is a smooth functional.  We denote by $!_{\delta}\sE$ the $c^\infty$-closure of the linear span of $\delta(\sE)$ in $(\sE^*)^\bot$.  

\begin{lem}\label{deltacomonad}
The endomorphism $!_{\delta}$ induces a comonad on $\tu{ConMod}(\s{C}^\infty_M)$.
\end{lem}

\begin{proof}
We have an inclusion $\Conv\ra\tu{TBorn}$ of closed symmetric monoidal categories which induces an inclusion $\tu{ConMod}(\sC^\infty_M)\ra\tu{TBMod}(\sC^\infty_M)$ of $\sC^\infty_M$-modules where 
\[  \tu{TBMod}(\sC^\infty_M):=\Mod_{\sC^\infty_M}(\Sh_{\tu{TBorn}}(M)). \] 
The left adjoint $\gamma:\tu{TBMod}(\sC^\infty_M)\ra\tu{ConMod}(\sC^\infty_M)$ of this inclusion is a composition of separation and completion functors.

We have a natural comultiplication map $\mu^\delta:!_\delta\ra !_\delta !_\delta$ which object-wise $\mu^\delta_\sE:!_\delta\sE\ra !_\delta!_\delta\sE$ extends linearly the map $\delta_s\mapsto\delta_{\delta_s}$ and applies the separation and completion functor $\gamma$.  The counit map $\epsilon^\delta:!_\delta\ra\id$ object-wise $\epsilon^\delta_\sE:!_\delta\sE\ra\sE$ extends linearly the map $\delta_s\mapsto s$ and applies the functor $\gamma$.  The commutativity of the relevant diagrams can be easily verified.
\end{proof}

We have a linear-non-linear adjunction
\[ \begin{tikzcd}
   \tu{ConMod}(\s{C}^\infty_M)_{!_\delta}\arrow[bend left]{r}{X}  
& \tu{ConMod}(\s{C}^\infty_M) \arrow[bend left]{l}{U}  
\end{tikzcd}  \]
and a symmetric monoidal comonad $!_\delta=X\circ U$ which we call the \textit{distributional comonad}.  The functor $X$ sends a $\sC^\infty_M$-module $\sE$ to the $c^\infty$-closure of the linear span of $\delta(\sE)$ and $U$ is a bijection on objects.    

\begin{prop}\label{kls}
There exists an equivalence
\[   \tu{ConMod}(\s{C}^\infty_M)_{!_\delta}\simeq \tu{ConMod}^{\tu{sm}}(\s{C}^\infty_M)  \]
of categories.
\end{prop}

\begin{proof}
The Dirac distributional density map is smooth.  It suffices to check the condition objectwise and so the result follows from Lemma~5.1 of \citep{BET}.
\end{proof}

Consider the sequent $!A\vdash B$ in differential linear logic with denotation $\llbracket-\rrbracket_M$ given by the functional $F:\llbracket !A\rrbracket_M=!_\delta\sE\ra\llbracket B\rrbracket_M=\sE'$ and the diagram
\[
\begin{tikzcd}
 & !_\delta\sE \arrow[d,"\epsilon^\delta_\sE"] \arrow[dr,"F"]  \\
* \arrow[ur, "\delta_s"] \arrow[r,"s"] & \sE   \arrow[r, "{F}^{sm}"] & \sE'
\end{tikzcd}
\]
of convenient $\sC^\infty_M$-modules. From Proposition~\ref{kls}, we have the smooth functional ${F}^{sm}:\sE\ra\sE'$ with ${F}^{sm}(s)\simeq F(\delta_{s})$ associated to $F$.  We also define a map $\bar{d}^\delta_{\sE}:\sE\ra !_\delta\sE$ for the distributional comonad, following \citep{BET}, by 
\[   \bar{d}^\delta_{\sE}(s)=\underset{h\ra 0}\lim\frac{\delta_{hs}-\delta_0}{h}  \]
where $s\in\sE$, $0$ is the zero section and $h$ the constant sheaf. 

\begin{thm}\label{didi}
The category of convenient $\sC^\infty_M$-modules with the distributional comonad $!_{\delta}$ and map $\bar{d}^\delta$ is a model for intuitionistic differential linear logic.
\end{thm}

\begin{proof}
The category of convenient vector spaces is locally presentable \citep{Wa} and closed symmetric monoidal \citep{KM}.  Sheaves with values in a locally presentable closed symmetric monoidal category themselves form a locally presentable closed symmetric monoidal category, as do modules over a commutative monoid object in such a category of sheaves \citep{Mes}.  Therefore the category of convenient $\sC^\infty_M$-modules is locally presentable closed symmetric monoidal.  It is moreover an additive, and therefore $\CMon$-enriched, symmetric monoidal category.

By Lemma~\ref{deltacomonad}, the functor $!_\delta$ is a comonad.  For each object $\sE$ in $\tu{ConMod}(\s{C}^\infty_M)$, we define a cocommutative comonoid object $(!_\delta\sE,c_\sE,e_\sE)$ using the maps $e_\sE:\delta_s\mapsto 1$ and
\[  c_\sE:\delta_s\mapsto \delta_s\otimes\delta_s,  \]
and then extending linearly and applying the separation and completion functor $\gamma$ (see the proof of Lemma~\ref{deltacomonad}).  Also, since the diagrams 
\[
\begin{tikzcd}
\delta_s\arrow[r, mapsto, "c_\sE"]  \arrow[d, mapsto, "\mu_\sE "]  & \delta_s\otimes\delta_s \arrow[d, mapsto, "\mu_\sE\otimes\mu_\sE "] \\
\delta_{\delta_s} \arrow[r, mapsto, "c_{!_\delta\sE}"]   & \delta_{\delta_s}\otimes\delta_{\delta_s}
\end{tikzcd} 
\hspace{1cm}
\begin{tikzcd}
\delta_s \arrow[r, mapsto, "\mu_\sE"] \arrow[rd, mapsto, "e_\sE"] & \delta_{\delta_s}  \arrow[d, mapsto, "e_{!_\delta\sE}"]\\
 & 1
\end{tikzcd}
\]
commute, $\mu_\sE:!_\delta\sE\ra!_\delta!_\delta\sE$ a morphism of comonoid objects in $\tu{ConMod}(\s{C}^\infty_M)$.  Let $\sE$ and $\sE'$ be convenient $\sC^\infty_M$-modules.  Then
\[     !_{\delta}(\sE\times\sE')\simeq !_{\delta}\sE\otimes !_{\delta}\sE' \]
is an isomorphism of sheaves by extending the fiberwise statement of Proposition~5.2.4 of \citep{FK} and Proposition~5.6 of \citep{BET}.  

It remains to show that the map $\bar{d}^\delta$ satisfies the conditions to be a codereliction.  Firstly, the diagram
\[
\begin{tikzcd}
\s{E}\arrow[r,"\bar{d}^\delta_\sE"]  \arrow[d,"F"]  & !_\delta\sE \arrow[d,"!_\delta(F)"] \\
\s{E}'\arrow[r,"\bar{d}^\delta_{\sE'}"]   & !_\delta\sE'
\end{tikzcd} 
\]
commutes since $\lim_{h\ra 0}\frac{\delta_{hF(s)}-\delta_0}{h}=\lim_{h\ra 0}\frac{\delta_{F(hs)}-\delta_{F(0)}}{h}$ owing to the property that $F$ is a morphism of $\s{C}^\infty_M$-modules (explicitly, $hs(x)=h(x)s(x)$ and $F(hs)(x)=h(x)F(s)(x))$.  Therefore, $\bar{d}^\delta$ is a natural transformation.  By Theorem~6 and Corollary~4 of \citep{BCLS}, it now suffices to show that the linear and chain rules of Definition~\ref{codereliction} are satisfied.  The left hand side of the linear rule $\epsilon_\sE\circ\bar{d}_\sE^\delta$ given by 
\[  s\mapsto\underset{h\ra 0}\lim\frac{\delta_{hs}-\delta_0}{h}\mapsto\underset{h\ra 0}\lim\big(\frac{1}{h}(hs-0)\big)=s          \]
coincides with the identity due to continuity of $\epsilon_\sE$.  The multiplication map of the monoid object in the bialgebra structure is given by $\bar{c}_{\sE}:\delta_s\otimes\delta_t\mapsto\delta_{s + t}$ and then extending linearly and applying $\gamma$.  Therefore, the left hand side $\mu_\sE\circ\bar{c}_\sE\circ(\bar{d}_\sE\otimes\id_{!_\delta\sE})$ of the chain rule gives 
\[  
s\otimes\delta_t\mapsto\left(\underset{h\ra 0}\lim\frac{\delta_{hs}-\delta_0}{h}\right)\otimes\delta_t\mapsto\underset{h\ra 0}\lim\frac{\delta_{hs + t}-\delta_t}{h}\mapsto\underset{h\ra 0}\lim\frac{\delta_{\delta_{hs + t}}-\delta_{\delta_t}}{h}    
\]
which corresponds to the right hand side $\bar{c}_{!_\delta\sE}\circ(\bar{d}_{!_\delta\sE}\otimes\mu_{!_\delta\sE})\circ (\bar{c}_\sE\otimes\id_{!_\delta\sE})\circ(\bar{d}_\sE\otimes c_\sE)$ by
\[  
s\otimes\delta_t\mapsto\left(\underset{h\ra 0}\lim\frac{\delta_{hs}-\delta_0}{h}\right)\otimes(\delta_t\otimes\delta_t)\mapsto\left(\underset{h\ra 0}\lim\frac{\delta_{hs + t}-\delta_t}{h}\right)\otimes\delta_t\mapsto  
\]
\[  
\left(\underset{h',h\ra 0}\lim\frac{\delta_{(\frac{h'}{h}(\delta_{hs + t}-\delta_t))} -\delta_{\delta_0}}{h'}\right)\otimes\delta_{\delta_t} \mapsto\underset{h',h\ra 0}\lim\frac{\delta_{(\frac{h'}{h}(\delta_{hs + t}-\delta_t)+\delta_t)} -\delta_{\delta_t}}{h'}             
\]
using associativity of the tensor product and then taking the limit $h=h'\rightarrow 0$ along the diagonal.
\end{proof}

We will call $\bar{d}^\delta$ the \textit{distributional codereliction}.  Let $F:!_\delta\sE\ra\sE'$ be a morphism in $\tu{ConMod}(\sC^\infty_M)$.  The deriving transformation $\overline{\partial}_\sE:\sE\otimes !_\delta\sE\ra !_\delta\sE$ is given by 
\[  
\overline{\partial}_\sE: t\otimes\delta_{s}\xmapsto{(\bar{d}^\delta_{\sE}\otimes\id)}\left( \underset{h\ra 0}\lim\frac{\delta_{ht}-\delta_0}{h}\right)\otimes\delta_{s}\xmapsto{\bar{c}_{\sE}}
\underset{h\ra 0}\lim\frac{\delta_{s+ht}-\delta_{s}}{h}
\]
and the derivative $\tu{d}{F}:=F\circ\overline{\partial}_{\sE}:\sE\otimes!_\delta\sE\ra\sE'$ of $F$ in $\tu{ConMod}(\sC^\infty_M)$ is 
\[  
\tu{d}{F}: t\otimes\delta_{s}\mapsto\underset{h\ra 0}\lim\frac{F(\delta_{s+ht})- F(\delta_s)}{h}   
 \]
for local sections $s,t\in\sE$.  Using the adjunction of Proposition~\ref{kls}, we have, by abuse of notation, an operator 
\[  
\tu{d}:\Hom_{\sC^\infty_M}^{\tu{sm}}(\sE,\sE')\ra\Hom_{\sC^\infty_M}(\sE,\uHom^{}_{\sC^\infty_M}(\sE,\sE'))  
\]
defined by
\[  
\tu{d}{F}^{sm}(s,t)=\underset{h\ra 0}\lim\frac{{F}^{sm}(s+ht)-{F}^{sm}(s)}{h} = \left.\frac{d}{dh}\right\vert_{h=0}{F}^{sm}(s+ht).  
\]
This derivative operator is linear and bounded and $\tu{d}{F}^{sm}(s,t)$ is the functional derivative at the section $s$ of $\sE$ in the direction of the section $t$.  When $\sE'=\sC^\infty_M$, another common notation for $\tu{d}{F}^{sm}(s,t)$ is 
\[  \tu{d}{F}^{sm}(s,t)= \int_U\frac{\delta {F}^{sm}}{\delta s}(x)t(x) dx  \]
for $U\subseteq M$.

%COMONAD COMPOSITION AND NON-LINEARITY
\section{Comonad composition and non-linearity}\label{ccomp}

In Section~\ref{csdc}, we have shown that the category of convenient $\sC^\infty_M$-modules is a model for intuitionistic differential linear logic using the distributional comonad $!_\delta$.  Combining this result with the extension of the model in Section~\ref{jetcom} to this same category, we obtain a compatible model based on composition with the infinite jet comonad, ie. these two comonads interact in a natural way so that their composition induces a model for intuitionistic differential linear logic.

Firstly, we update the finite jet functor by lifting it to an endofunctor
$   !_{j^r}: \tu{ConMod}(\s{C}^\infty_M)\ra\tu{ConMod}(\s{C}^\infty_M)  $
and leverage the convenient structure to define
\[   \sJ(\sE):=\underset{r\in\mathbb{N}}{\lim}(\s{J}^r(\sE))  \]
as a genuine limit in $\tu{ConMod}(\s{C}^\infty_M)$.  The infinite prolongation thus induces an endofunctor
\[  !_j: \tu{ConMod}(\s{C}^\infty_M)\ra\tu{ConMod}(\s{C}^\infty_M)  \]
on the category of convenient $\sC^\infty_M$-modules.  The following result is clear from Lemma~\ref{comonad}.

\begin{cor}
The endomorphism $!_j$ is a comonad on $\tu{ConMod}(\s{C}^\infty_M)$.
\end{cor}

The category $\tu{ConMod}(\s{C}^\infty_M)$ is endowed with a cocartesian monoidal structure with monoidal product $\oplus$ and unit $0$.  We have a linear-non-linear adjunction
\[ \begin{tikzcd}
   \tu{ConMod}(\s{C}^\infty_M)_{!_j}\arrow[bend left]{r}{X}  
& \tu{ConMod}(\s{C}^\infty_M) \arrow[bend left]{l}{U}  
\end{tikzcd}  \]
and a symmetric monoidal comonad $!_j=X\circ U$ which we call the jet comonad.  Here $X$ sends a $\sC^\infty_M$-module $\sE$ to $\sJ(\sE)$ and the right adjoint $U$ is an object bijection.  The jet codereliction $\bar{d}^j$ extends to a natural transformation on $\tu{ConMod}(\sC^\infty_M)$.  A corollary of Theorem~\ref{thm1} is now the following.

\begin{cor}
The category of convenient $\sC^\infty_M$-modules with the jet comonad $!_j$ is a symmetric monoidal storage category for the cocartesian monoidal structure.
\end{cor}

Owing to the discussion in Section~\ref{jetcom}, we have an isomorphism
\[  \Hom_{\tu{ConMod}^{\tu{sm}}(\sC^\infty_M)}(\sJ(\sE),\sE')\simeq\tu{Diff}^{\,\tu{sm}}(\sE,\sE')  \]
where the right hand side denotes the set of \textit{smooth partial differential operators}.  

We now define a number of different functionals on the space of jets of sections of a vector bundle.

\begin{dfn}
Let $E$ be a vector bundle on $M$.  A \textit{local linear functional} on $\sE$ is an element of the continuous linear dual $(!_j\sE)^\bot$.  A \textit{local smooth functional} on $\sE$ is an element of the smooth dual $(!_j\sE)^*$.  
\end{dfn}

Local smooth functionals are also called \textit{Lagrangians} in certain applications.  Lagrangians given by formal power series are particularly important in the study of perturbative classical and quantum field theories.  This is demonstrated in the following example.

\begin{ex}\label{laf}
Building on Example~\ref{fps}, the algebra of formal power series of local linear functionals is given by
\[   \s{O}^{\tu{loc}}_{\s{E}}:=\widehat{\Sym}_{\sC^\infty_M}(!_j\s{E}^\bot)   \]
elements of which will be called \textit{Lagrangian densities}.  More explicitly, we identify the $n$th component of a Lagrangian density on $M$ with a compactly supported distributional section of the bundle $(J(E)^\forall)^{\boxtimes n}$ on $M^n$.  Since local linear functionals depend only on the local nature of a section $s$ at each point, ie. its jet, then we can interpret its $n$th component as a finite sum of densities of the form
$  (D_1s)(D_2s)\ldots(D_ns)\tu{d}\Omega $
where each $D_i:\sE\ra\sC^\infty_M$ is a differential operator.  The natural inclusion 
\[     \iota_U:\s{O}^{\tu{loc}}_{\s{E}}(U)\ra\widehat{\sO}_{\sE}(U) \]
given by integration $\iota_U(\cL):s\mapsto\int_U\cL(s)$ defines the \textit{action} $S_U:=\iota_U(\cL):\sE(U)\ra\mathbb{R}$ of the Lagrangian distributional density $\cL$.  
\end{ex} 

\begin{rmk}
Note that the section $s$ in Example~\ref{laf} should be nilpotent since in most cases, the infinite sum will not converge.  Alternatively, we could define a Lagrangian density to be an element $\cL$ in $\s{O}^{\tu{loc}}_{\s{E}}$ which factors through $\prod_{n=0}^r((!_j\s{E}^\bot)^{\otimes n})_{S_n}$ for some finite $r$.
\end{rmk}

From Example~\ref{laf}, a Lagrangian sends a section $s$ in $\s{E}(U)$ to a formal power series in these variables, a density which, when evaluated on a point in $U$ depends only on the infinite jet at that point.   

We endow $!_j\sE$ with its canonical $\sD^\infty_M$-module structure.  This is the canonical flat connection given by the Cartan distribution of Section~\ref{jetcom}.  Then the convenient $\sC^\infty_M$-module
$!_j\sE^\bot=\uHom_{\sC^\infty_M}(!_j\sE,\sC^\infty_M)$ has a canonical $\sD^\infty_M$-module structure.  Therefore, a local functional is a $\sD^\infty_M$-module.  Again, every element $L$ of this module takes a section $s$ of $\sE(U)$ and returns a smooth function $L(s)$ in $\sC^\infty_M(U)$ with the property that $L(s)(x)$ depends only on the $\infty$-jet of $s$ at $x\in U$.

Now that the jet comonad is understood as a comonad on the category of convenient $\sC^\infty_M$-modules, we combine this result with the distributional comonad of Section~\ref{csdc}.  Pre-composition with $!_j$ gives the module $!_{j\delta}\sE:= !_\delta\circ !_j\sE$.  So $!_{j\delta}\sE$ is the $c^\infty$-closure of the linear span of $\delta(!_j\sE)$ in $(!_j\sE^*)^\bot$.  The two comonads interact in the expected way.   

\begin{lem}\label{jdstorage}
The composite comonad $!_{j\delta}$ is a storage comonad.  
\end{lem}

\begin{proof}
There is a canonical distributive law of $!_j$ over $!_\delta$ since, by definition, operators act on distributions as
$   \< s, j(\delta) \> := \<j(s),\delta\> $
and therefore $!_{\delta j}\simeq !_{j\delta}$ is an isomorphism of comonads.  The result now follows from Lemma~\ref{dchase}.
\end{proof}

The comonad $!_{j\delta}$ will be called the \textit{jet-distributional comonad}.  Now consider the map $\bar{d}^{j\delta}:\id\ra !_{j\delta}$ given by
\[   \bar{d}^{j\delta}_{\sE}(s)=\underset{h\ra 0}\lim\frac{\delta_{h(j(s))}-\delta_0}{h}  \]
for $s\in\sE$.

\begin{thm}\label{jdt}
The category of convenient $\sC^\infty_M$-modules with the jet-distributional comonad $!_{j\delta}$ and map $\bar{d}^{j\delta}$ is a model for intuitionistic differential linear logic.
\end{thm}

\begin{proof}
By Lemma~\ref{jdstorage}, the jet-distributional comonad is a storage comonad.  The remainder of the proof is obtained by applying the corresponding proof in Theorem~\ref{didi} to the convenient $\sC^\infty_M$-module $!_j\sE$.
\end{proof}

The map $\bar{d}^{j\delta}:\id\ra !_{j\delta}$ will be called the \textit{jet-distributional codereliction}.  We have a linear-non-linear adjunction
\[ \begin{tikzcd}
   \tu{ConMod}(\s{C}^\infty_M)_{!_{j\delta}}\arrow[bend left]{r}{X}  
& \tu{ConMod}(\s{C}^\infty_M) \arrow[bend left]{l}{U}  
\end{tikzcd}  \]
where the functor $X$ sends an object $\sE$ to the $c^\infty$-closure of the linear span of $\delta(\sJ(\sE))$ and the functor $U$ is a bijection on objects.  Objects on the left hand side are convenient vector bundles and whose morphisms, owing to Proposition~\ref{kls}, include \textit{non-linear} partial differential operators $\widehat{F}^{sm}:\sE\ra\sE'$.  Indeed, let $F:!_{j\delta}\sE\ra\sE'$ be a morphism of $\sC^\infty_M$-modules and consider the diagram
\[
\begin{tikzcd}
                & !_{j\delta}\sE \arrow[d,"\epsilon^\delta_{!_j\sE}"] \arrow[ddr,bend left=20,near start,"F"]  \\
                          &   !_j\sE  \arrow[d,"\epsilon^j_{\sE}"] \arrow[dr,near start,"{F}^{sm}"]  \\
* \arrow[ur,near end,"j(s)"] \arrow[uur,bend left=20,"\delta_{j(s)}"] \arrow[r,"s"]   &\sE \arrow[r, dashrightarrow,"\widehat{F}^{sm}"] & \sE'
\end{tikzcd}
\]
in $\tu{ConMod}(\s{C}^\infty_M)$.  We have $F(\delta_{j(s)})\simeq F^{sm}(j(s))\simeq\widehat{F}^{sm}(s)$.  Moreover, taking advantage of the closed structure and using the notation of linear logic, we have a commutative diagram
\[
\begin{tikzcd}
!_{j\delta}\sE\multimap\sE'  \arrow[r,"\circ j"]  \arrow[d,"\circ\overline{d}^\delta_\sE"]  & !_\delta\sE\multimap\sE' \arrow[d,"\circ\overline{d}^\delta_\sE"] \\
!_j\sE\multimap\sE'    \arrow[r,"\circ j"]   & \sE\multimap\sE'
\end{tikzcd} 
\]
sending the convenient $\sC^\infty_M$-module of smooth local functionals to the convenient $\sC^\infty_M$-module of linear functionals.

The deriving transformation $\overline{\partial}_\sE:\sE\otimes!_{j\delta}\sE\ra!_{j\delta}\sE$ is defined as the composite 
\[  
\overline{\partial}_\sE: t\otimes\delta_{j(s)}\xmapsto{(\bar{d}^{j\delta}_{\sE}\otimes 1)}\left(\underset{h\ra 0}\lim\frac{\delta_{h(j(t))}-\delta_0}{h}\right)\otimes\delta_{(j(s))}\xmapsto{\bar{c}_{\sE}}\underset{h\ra 0}\lim\frac{\delta_{(j(s)+h(j(t)))}-\delta_{j(s)}}{h}  
\]
and the derivative $\tu{d}F:=F\circ\overline{\partial}_\sE:\sE\otimes!_{j\delta}\sE\ra\sE'$ of $F:!_{j\delta}\sE\ra\sE'$ in $\tu{ConMod}(\s{C}^\infty_M)$ is given as
\[
\tu{d}F: t\otimes\delta_{j(s)}\mapsto\underset{h\ra 0}\lim\frac{F(\delta_{j(s)+hj(t)})- F(\delta_{j(s)})}{h}. 
\]

By abuse of notation, we have an operator on smooth differential operators
\[  \tu{d}:\tu{Diff}^{\,\tu{sm}}(\sE,\sE')\ra\Hom(\sE,\mathscr{D}_M(\sE,\sE'))  \]
defined by
\[   \tu{d}\widehat{F}^{sm}(s,t)=\underset{h\ra 0}\lim\frac{\widehat{F}^{sm}(s+ht)-\widehat{F}^{sm}(s)}{h} =  \left.\frac{d}{dh}\right\vert_{h=0}\widehat{F}^{sm}\big(s+ht)  \]
which is linear and bounded, ie. $\tu{d}\widehat{F}^{sm}(s,t)$ is the deriviative of the smooth local functional $\widehat{F}^{sm}$ at $s$ in the direction $t=ds$.

When our sheaf is finite dimensional we have the following more explicit description of non-linear local functionals.

\begin{ex}
Let $\sE$ be a finite dimensional convenient $\sC^\infty_M$-module.  Then there exists an isomorphism
\[   !_{j^r\delta}\sE\simeq (!_{j^r}\sE^*)^\bot  \]
of convenient $\sC^\infty_M$-modules.  This can be deduced from Corollary~5.1.8 of \citep{FK}.
\end{ex}

When our smooth functionals are given by formal power series, we also have a more explicit description.  We endow $\s{D}ens_M$ with its right $\sD_M^\infty$-module structure.  Then a local density on $U\subseteq M$ with respect to $\sE$ is an element $\omega_U\otimes L$ in the space
\[      \sO^{\loc}_{\sE(U)}\simeq\s{D}ens_M(U)\otimes_{\s{C}^\infty_M(U)}\widehat{\Sym}_{\mathbb{R}}(!_j\overline{\sE}(U)^\vee)  \]
with its canonical $\sD^\infty_M(U)$-module structure where $!_j\sE(U)^\vee=\uHom(!_j\sE(U),\sC^\infty_M(U))$.  In words, the element $\omega_U\otimes L$ sends a section $s$ in $\sE(U)$ to a distributional density $\omega_U\otimes L(s)$ on $U$ such that $(\omega_U\otimes L(s))(x)$ depends only on the $\infty$-jet of $s$ at $x\in U$.  

We obtain a sheaf $\sO^{\loc}_{\sE}$ on $M$ which is moreover a $\sC^\infty_M$-module.  It is a subsheaf $\sO^\loc_{\sE}\subset!_j\sE^*$ and the $c^\infty$-closure of the linear span of $\delta(!_j\sE)$ in $(!_j\sE^*)^\bot$ factors through $(\sO^\loc_{\sE})^\bot$.  This is a subspace of $!_{j\delta}\sE(U)$.  This restricted delta distribution $\delta$ sends $j(s)$ to 
\[  \delta_{j(s)}:\omega_U\otimes L\mapsto\int_U L(j(s))\omega_U  \]
where $L$ is a formal power series.  In local coordinates on $U\subseteq M$, and using integration by parts, we have 
\[    \tu{d}S(s)=\tu{d}\int_U L(j(s))\omega_U=\int_U el_\alpha(L)\tu{d}s^\alpha\wedge\omega_U + D_\alpha V^\alpha  \]
for some total derivative $D_\alpha V^\alpha$ where $el_\alpha=\sum_I (-D)_{I} \frac{\partial}{\partial u_{I}^{\alpha}}$ is the Euler-Lagrange operator.  Here $(-D)_I=(-D_{i_1})(-D_{i_2})\ldots$ for the multi-index $I$.  

Therefore, a Lagrangian is only defined up to a total derivative for compactly supported sections.  This can be exploited by forming the tensor product
\[      \sO^{\tu{red}}_{\sE}\simeq\s{D}ens_M\otimes_{\s{D}^\infty_M}\widehat{\Sym}_{\mathbb{R}}(!_j\overline{\sE}^\vee)  \]
over $\sD^\infty_M$.  Therefore $\tu{d}S(s)=0$ if and only if the section $s$ satisfies the Euler-Lagrange equations $el_\alpha(L(j(s)))=0$.  Symmetries of the action can also be interpreted as vector fields on the jet bundle (cf. Remark~\ref{symmetries}).  

We end by giving a concrete application of this construction.

\begin{ex}[Free and interacting scalar fields]
Fix a $n$-dimensional compact Riemannian manifold $(M,g)$.  Consider the sheaf $\sE$ of sections of the trivial bundle $\pi:E:=M\times\mathbb{R}\ra M$, ie. $\sE$ is simply the sheaf $\sC^\infty_M$ of smooth functions on $M$.  There exists an isomorphism $(!_j\sC^\infty_M)^\bot\simeq\sD_M^\infty$ of (left) $\sD^\infty_M$-modules.  The elements $\widehat{L}$ in $(!_j\sC^\infty_M)^\bot$ are spanned by elements of the form $\phi(x_i)\partial_{I}$ for $x_i\in M$ and a partial differential operator $\partial_I$ depending on a multi-index $I$.  

Let $\phi\in\sC^\infty_M(M)$ be a scalar field.  We consider the special forms of $\widehat{L}$ given by
\[   \widehat{L}(j(\phi))=\phi D\phi,  \quad\quad \widehat{L}(j(\phi))=\phi D\phi + \eta\phi,   \quad\quad {L}(\delta_{j(\phi)})=\phi D\phi  + V(\phi)   \] 
where the Laplacian $D$ is the differential operator $D:\s{C}_M^\infty(M)\ra\s{C}_M^\infty(M)$ sending $\phi$ to $\Delta_g\phi$ and the density is the canonical volume form.  The first two functionals are linear local functionals whereas the last functional is merely smooth in general.  The functional derivative of the local action functional $S$ associated to $L$ is   
\[  \tu{d}S(\phi) =\int_M\tu{d}\widehat{L}^{sm}(\phi)\omega_M =\int_M el(\widehat{L}^{sm}(\phi))\tu{d}\phi\,vol_g  \]
where, for local coordinates $(x^i,\phi,\phi_i)$, the Euler-Lagrange equations are
\[ 
el(\widehat{L}^{sm}(\phi))= \frac{\partial\widehat{L}^{sm}}{\partial\phi}-\frac{\partial}{\partial x^i}\left(\frac{\partial\widehat{L}^{sm}}{\partial\phi_i}\right).
\]
The principle of least action $\tu{d}S=0$, or equivalently $el(\widehat{L}^{sm}(\phi))=0$, leads to the partial differential equations 
\[  \Delta_g\phi = 0,   \quad\quad  \Delta_g\phi = \eta , \quad\quad  \Delta_g\phi = -V'(\phi)   \]
which are the Laplace, Poisson and non-linear Poisson equation respectively.  These define a vector subbundle, affine subbundle and fibered submanifold of $J^2(M\times\mathbb{R})$ respectively.  

To see this, let $(x^i,\phi,\phi_i,\phi_{ij})$ be coordinates on $J^2(M\times\mathbb{R})$ and consider the function $f(x^i,\phi,\phi_i,\phi_{ij})=\sum_{1\leq i\leq n}\phi_{ii}$ on $J^2(M\times\mathbb{R})$.  The preimage of $0$, $\eta(x)$ and $-V'(\phi)$ with respect to $f$ define a fibered submanifold ${H}^2\subseteq J^2(M\times\mathbb{R})$.  Taking the infinite prolongation of the equation $H^2$ we obtain the equation $H$ which, assuming $H^2$ is regular, is a pro-ind vector subbundle, pro-ind affine subbundle and pro-ind fibered submanifold of $J(M\times\mathbb{R})$ respectively.  A local section $\phi$ of $\pi:M\times\mathbb{R}\ra M$ is a solution of these equations if and only if $j^2\phi(x_i)\in\s{H}^2\simeq\s{H}$. 
\end{ex}

%CONCLUSION
\section{Conclusion}

We have shown that the category of convenient sheaves is a model for intuitionistic differential linear logic.  Using the jet comonad for the exponential modality gives an interpretation of linear differential operators, and hence linear partial differential equations, in linear logic.  Alternatively, using the distributional comonad for the exponential gives an interpretation of smooth morphisms between objects in these categories.  Composing these comonads provides an interpretation of non-linear differential operators and the variational calculus of smooth local functionals within linear logic.

Some interesting questions remain open.  The most pressing item is to elucidate the internal logic of the model in order to provide a computational interpretation of its structure within differential $\lambda$-calculus.  Indeed, the Kleisli category of a model for intuitionistic differential linear logic is a cartesian closed differential category and it is these categories, introduced in \citep{BEM} as \textit{differential $\lambda$-categories}, that are models of the simply typed differential $\lambda$-calculus.  See \citep{Man,BCS4} for more details.

Other interesting questions include the extension to classical differential linear logic \citep{Gi2}, the exploration of antiderivatives and integration from the perspective of \citep{Eh3} and the application of reverse-mode differentiation from \citep{CCG}.  Finally, we would like an expansion of the category of vector bundles to include ``non-smooth" structures.  This requires the introduction of tools from synthetic and derived differential geometry.  These more elaborate structures are needed to make sense of non-linear partial differential equations and their moduli space of solutions within the context of models of differential linear logic.

%THE BIBLIOGRAPHY..........................................................................................................................................................................................................
\bibliographystyle{abbrvnat}
\bibliography{Jets}

\end{document}